\documentclass[journal]{IEEEtran}

\ifCLASSOPTIONcompsoc
  \usepackage[nocompress]{cite}
\else
  \usepackage{cite}
\fi

\usepackage[T1]{fontenc}
\usepackage{enumitem}
\usepackage{amsmath}
\usepackage{amssymb}
\usepackage{amsthm}
\usepackage{mathdots}
\usepackage{tikz}
\usetikzlibrary{shapes,arrows,positioning,decorations.markings}
\usepackage{multirow,array}
\usepackage{tabularx}
\usepackage{bm}
\usepackage{graphicx}
\usepackage{subfigure}
\usepackage{algorithm}
\usepackage{algorithmic}
\usepackage{textcomp}
\usepackage{mdframed}
\usepackage{color}

\definecolor{orange}{rgb}{1, .36, .08}

\newtheorem{thm}{Theorem}

\newtheorem{problem}{Problem}
\newtheorem{example}{Example}

\newcommand{\ignore}[1]{}

\newcommand{\G}{{\bf G}}
\newcommand{\F}{{\bf F}}
\newcommand{\U}{{\bf U}}

\newcommand{\W}{{\bf W}}
\newcommand{\pfr}[1]{\mathcal{#1}}

\newcommand{\secref}[1]{\mbox{Sec.{ }\ref{sec:#1}}}

\newcommand{\seclabel}[1]{\label{sec:#1}}

\newcommand{\sem}[1]{[\![{#1}]\!]}

\def\C{\mathcal{C}}

\begin{document}


\title{Stochastic Assume-Guarantee Contracts for \\ Cyber-Physical System Design Under  \\ Probabilistic Requirements}

\author{\IEEEauthorblockN{Jiwei Li$^1$,
	Pierluigi Nuzzo$^2$,
	Alberto Sangiovanni-Vincentelli$^3$,
	Yugeng Xi$^1$,
	Dewei Li$^1$} \\
\IEEEauthorblockA{$^1$ Department of Automation, Shanghai Jiao Tong University. Email: adanos@126.com, \{ygxi,dwli\}@sjtu.edu.cn\\
$^2$ Department of Electrical Engineering, University of Southern California, Los Angeles. Email: nuzzo@usc.edu \\
$^3$ EECS Department, University of California, Berkeley. Email: alberto@eecs.berkeley.edu
}
}


\maketitle

\begin{abstract}
%
%
We develop an assume-guarantee contract framework for the design of cyber-physical systems, modeled as closed-loop control systems, under probabilistic requirements.
We use a variant of signal temporal logic, namely, Stochastic Signal Temporal Logic (StSTL) to specify system behaviors as well as contract assumptions and guarantees, thus enabling automatic reasoning about requirements of stochastic systems. Given 
a stochastic linear system representation and a set of requirements captured by bounded StSTL contracts, we propose algorithms that can check contract compatibility, consistency, and refinement, and generate a controller to guarantee that a contract is satisfied, following a stochastic model predictive control approach. Our algorithms leverage encodings of the verification and control synthesis tasks into mixed integer optimization problems, and conservative approximations of probabilistic constraints that produce both sound and tractable problem formulations. We illustrate the effectiveness of our approach on a few examples, including the design of embedded controllers for aircraft power distribution networks.
\end{abstract}

\section{Introduction}

Large and complex Cyber-Physical Systems (CPSs), such as intelligent buildings, transportation, and energy systems, cannot be designed in a monolithic manner. Instead, designers use hierarchical and compositional methods, which allow assembling a large and complex system from smaller and simpler components, such as pre-defined library blocks. Contract-based design is emerging as a unifying formal compositional paradigm for CPS design and has been demonstrated on several applications~\cite{Benveniste2013,Nuzzo15b}. It supports requirement engineering by providing formalisms and mechanisms for early detection of integration errors, for example, by checking compatibility between components locally, before performing expensive, global system verification tasks. 
%
%
%
However, while a number of contract and interface theories have appeared to support deterministic system models~\cite{Benveniste08,Alfaro01_2}, the development of contract frameworks for stochastic systems under probabilistic requirements is still in its infancy. 

Deterministic approaches fall short of accurately capturing those aspects of practical systems that are subject to variability (e.g., due to manufacturing tolerances, usage, and faults), noise, or model uncertainties. While trying to meet the specifications over the entire space of uncertain behaviors,  they tend to produce worst-case designs that are overly conservative.
Moreover, several design requirements in practical applications cannot be rigidly defined, and would be better expressed as probabilistic constraints, e.g., to formally capture that 
``the room temperature in a building shall be in a comfort region with a confidence level larger than 80\% at any time during a day.''
Providing support for reasoning about probabilistic behaviors and for the development of robust design techniques that can avoid over-design is, therefore, crucial. This need becomes increasingly more compelling as a broad number of safety-critical systems, such as autonomous vehicles, uses machine learning and statistical sensor fusion algorithms to infer information from the external world.

An obstacle to the development of stochastic contract frameworks and their adoption in system design stems from the computational complexity of the main verification and synthesis tasks for stochastic systems (see, for example, \cite{KNP07a,KNP11}), which are needed to perform concrete computations with contracts. 
%
%
%
A few proposals toward a specification and contract theory for stochastic systems
have recently appeared, e.g., based on 
Interactive Markov Chains~\cite{gossler12}, Constraint Markov Chains~\cite{Caillaud10}, and Abstract Probabilistic Automata~\cite{delahaye2011abstract,delahaye2011apac}. However, these frameworks mostly use contract representations based on automata, which  are more suitable to reason about discrete-state discrete-time system abstractions. They  tend to favor an imperative specification style, and may show poor scalability when applied to hybrid systems. 

A declarative specification style is often deemed as more practical for system-level requirement specification and validation, since it retains a better correspondence between informal requirements and formal statements.  
In this paper, we develop an A/G contract framework for automated design of CPSs  modeled as closed-loop control systems under probabilistic requirements. We aim to identify formalisms for contract representation and manipulation that effectively trade expressiveness with tractability: (i) they are rich enough to represent \emph{hybrid system behaviors} using a \emph{declarative style}; (ii) they are amenable to algorithms for \emph{efficient computation} of contract operations and relations.

We address these challenges by leveraging an extension of Signal Temporal Logic (STL)~\cite{MalerN04}, namely, Stochastic Signal Temporal Logic (StSTL), to support the specification of probabilistic constraints in the contract assumptions and guarantees. We show that the main verification tasks for bounded StSTL contracts on stochastic linear systems, i.e., compatibility, consistency, and refinement checking, as well as the synthesis of stochastic Model Predictive Control (MPC) strategies can all be translated into mixed integer programs (MIPs) which can be efficiently solved by state-of-the-art tools. Since probabilistic constraints on stochastic systems cannot be expressed in closed analytic form except for a small set of stochastic models~\cite{nemirovski2006convex}, we propose conservative approximations to provide optimization problem formulations that are both sound and tractable. We illustrate the effectiveness of our approach with a few examples, including the synthesis of controllers for an aircraft electric power distribution system.



\emph{\textbf{Related Work.}} A generic assume-guarantee (A/G) contract framework for probabilistic systems that can also capture reliability and availability properties using a declarative style has been recently proposed~\cite{delahaye2010probabilistic}. Our work differs from this effort, since it is not based on a probabilistic notion of contract satisfiability. In our approach, probabilistic constraints appear, instead, as predicates in the contract assumptions and guarantees. 

We express assumptions and guarantees using StSTL, which is an extension of STL~\cite{MalerN04}. STL was proposed for the specification of properties of continuous-time real-valued signals and has been previously used in CPS design~\cite{Nuzzo15b}. 
%
%
A few probabilistic extensions of temporal logics have been proposed over the years to express properties of stochastic systems.
Among these, Probabilistic Computation Tree Logic (PCTL) was introduced
to expresses properties over the realizations (paths) of finite-state Markov
chains and Markov decision processes~\cite{hansson1994logic} by extending
the Computation Tree Logic (CTL)~\cite{clarke1986automatic}. While PCTL can reason about global system executions and uncertainties about the times of occurrence of certain events, 
certain applications are rather concerned with capturing the uncertainty on the value of a signal at a certain time. This is the case, for instance, in the deployment of stochastic MPC schemes in different domains. 
By using StSTL, we can express requirements where uncertainty is restricted to probabilistic predicates and does not involve temporal operators.
While being expressive enough to cover the applications of interest, this restriction is also convenient, since it allows directly translating design and verification problems into optimization and feasibility problems with chance (probabilistic) constraints that can be efficiently solved using off-the-shelf tools. 


Closely related to StSTL, Probabilistic Signal Temporal Logic (PrSTL)~\cite{sadigh2016} has been recently proposed
to specify properties and design controllers for deterministic systems in uncertain
environments, captured by Gaussian stochastic processes.
%
Our work is different since it focuses on developing a comprehensive contract framework that supports both verification and control synthesis tasks. Our framework can reason about a broader class of systems, including linear systems with additive and control-dependent noise and Markovian jump linear systems.
Moreover, it supports non-Gaussian probabilistic constraints 
that cannot be captured in closed analytic form, by formulating encodings of synthesis and verification tasks that can produce sound and efficient approximations. 
%
%


\section{Preliminaries}
\label{sec:background}

As we aim to extend the \emph{Assume-Guarantee (A/G) 	contract} framework~\cite{Benveniste2013} to stochastic systems, we start by providing some background on A/G contracts and Stochastic Signal Temporal Logic (StSTL). 

\subsection{Assume-Guarantee Contracts: An Overview}
\seclabel{agc}

The notion of contracts originates from 
\emph{assume-guarantee reasoning}~\cite{Clark99}, which has been known for a
long time as a hardware and software verification technique. However,  its adoption in the context of  reactive systems, i.e., systems that maintain an ongoing interaction with their environment, such as CPSs, has been advocated only recently~\cite{Benveniste2013,Sangiovanni-Vincentelli2012a}.

We provide an overview of A/G contracts starting with a generic
representation of a component. We associate to it a set of properties that
the component satisfies, expressed with contracts. The contracts will be used to
verify the correctness of the composition and of the refinements.
A component is an element of a design,
characterized by a set of \emph{variables} (input or output), a set of  \emph{ports} (input or output), and a set of \emph{behaviors} over its variables and
ports. Components can be connected together by sharing certain ports under
constraints on the values of certain variables.
Behaviors are generic and could be continuous functions that result from solving
differential equations, or sequences of values or events recognized by an
automaton. To simplify, we use the same term ``variables'' to
denote both component variables and ports.
We use $\sem{M}$ to denote the
set of behaviors of component $M$.

A \emph{contract} $C$ for a component $M$ is a triple $(V, A, G)$, where
$V$ is the set of component variables, and $A$ and $G$ are sets of behaviors
over $V$~\cite{Benveniste08}. $A$ represents the \emph{assumptions} that $M$
makes on its environment, and $G$ represents the \emph{guarantees} provided by
$M$ under the environment assumptions.
A component $M$ satisfies a contract $C$ whenever $M$ and $C$ are
defined over the same set of variables, and all the behaviors of $M$ are
\emph{contained} in the guarantees of $C$ once they are composed (i.e.,
intersected) with the assumptions, that is, when $\sem{M} \cap A \subseteq G$. We
denote this \emph{satisfaction} relation by writing  $M \models C$, and we
say that $M$ is an \emph{implementation} of $C$.
However, a component $E$ can also be associated to a contract $C$ as an
\emph{environment}. We say that $E$ is a legal environment of $C$, and
write $E \models_E C$, whenever $E$ and $C$ have the same variables
and $\sem{E} \subseteq A$.

A contract $C = (V, A, G)$ is in \emph{canonical form} if the \emph{union}
of its guarantees $G$ and the complement of its assumptions $A$ is coincident
with $G$, i.e., $G = G \cup \overline{A}$, where $\overline{A}$ is the
complement of $A$. Any contract $C$ can be turned into a contract in
canonical form $C'$ by taking $A'=A$ and $G' = G \cup \overline{A}$.
We observe that
$C$ and $C'$ possess identical sets of environments and implementations.
Such two contracts $C$ and
$C'$ are then \emph{equivalent}.  Because of this equivalence, in what follows, we assume that all
contracts are in canonical form.

A contract is \emph{consistent} when the set of implementations satisfying it is
not empty, i.e., it is feasible to develop implementations for it. This amounts to verifying that $G \neq \emptyset$, where
$\emptyset$ denotes the empty set. Let $M$ be any implementation; then $C$ is \emph{compatible} if there exists a legal
environment $E$ for $M$, i.e., if and only if $A \neq \emptyset$. The intent is
that a component satisfying contract $C$ can only be used in the context
of a compatible environment.

Contracts can be combined according to
different rules. \emph{Composition} ($\otimes$) of contracts can be used to construct complex global 
contracts out of simpler local ones.
%
Let $C_1$ and $C_2$ be contracts over the same set of variables $V$. 
Reasoning on the compatibility and consistency of the composite contract $C_1 \otimes C_2$ can then be used to assess whether there exist components $M_1$ and $M_2$ such that their composition is valid, even if the full implementation of $M_1$ and $M_2$ is not available.

To reason about consistency between different abstraction layers in a design,
contracts can be ordered by establishing a \emph{refinement} relation. We say that $C$ refines $C'$, written
$C \preceq C'$, if and only if $A \supseteq A'$ and $G \subseteq G'$. Refinement amounts to
relaxing assumptions and reinforcing guarantees.
Clearly, if $M \models C$ and $C \preceq C'$, then $M \models
C'$. On the other hand, if $E \models_E C'$, then $E \models_E
C$. In other words, contract $C$ refines $C'$, if $C$
admits less implementations than $C'$, but more legal environments than
$C'$. We can then replace $C'$ with $C$.

Finally, to combine multiple requirements on the same component that need to be
satisfied simultaneously, the \emph{conjunction} ($\wedge$) of contracts can
also be defined so that, if a component $M$ satisfies the conjunction of $\C_1$
and $\C_2$, i.e., $M \models C_1 \wedge C_2$, then it also satisfies
each of them independently, i.e., $M \models C_1$ and $M \models
C_2$.
We refer the reader to the literature~\cite{Benveniste2013} for the formal definitions and mathematical expressions of contract composition and conjunction. 
In the following, we provide concrete representations of some of these operations and relations using operations on StSTL formulas.

\subsection{Stochastic Signal Temporal Logic (StSTL)}\label{sec:StSTL}

We use StSTL to formalize requirements for discrete-time stochastic system and express both contract assumptions and guarantees. However, similarly to STL, StSTL also extends to continuous-time systems. 

\begin{figure}[t]
	\centering
	\includegraphics[width=0.3\textwidth]{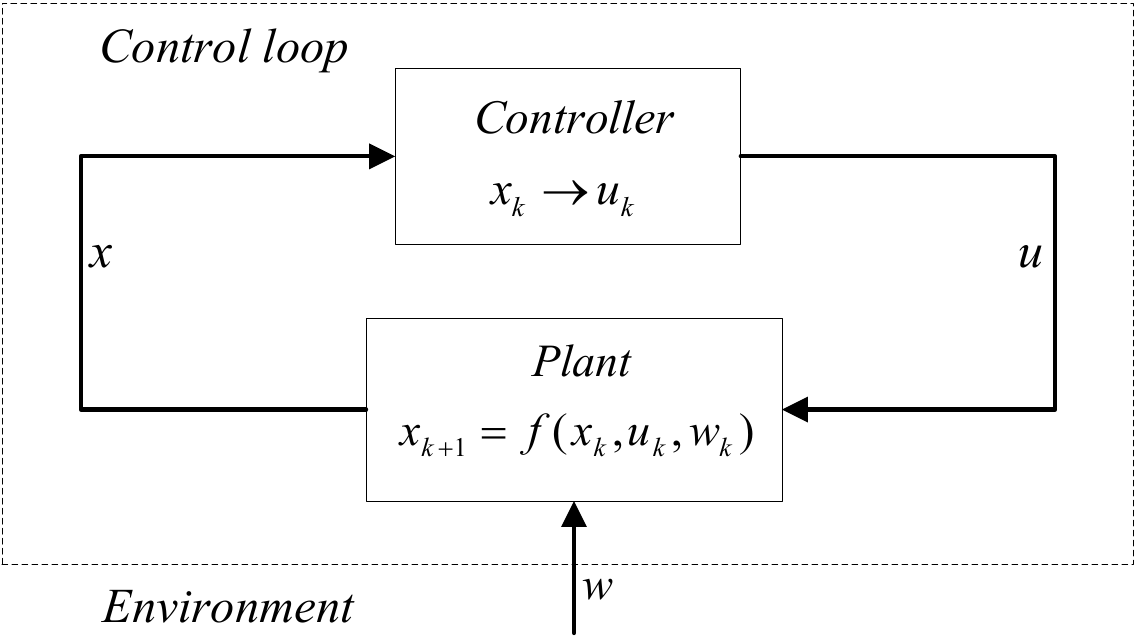}
	\caption{Components in the control loop and their interactions.}
	\label{fig:sys_structure}
\end{figure}

\textbf{Stochastic System.} We consider a discrete-time stochastic system in a classic closed-loop control configuration as shown in Fig.~\ref{fig:sys_structure}. The system dynamics are given by
\begin{equation}\label{eq:sys}
x_0 = \bar{x}_0, \quad x_{k+1} = f(x_k,u_k,w_k), \quad k=0,1,\ldots
\end{equation}
where $f$ is an arbitrary measurable function~\cite{Durrett10}, $x_k\in\mathbb{R}^{n_x}$ is the system state, $\bar{x}_0$ is the initial state, $u_k\in\mathbb{R}^{n_u}$ is
the (control) input, and $\{w_{k}\}_{k=0}^\infty$ is a random process on a complete
probability space, which we denote as $(\Omega, \mathcal{F}, \mathcal{P})$, using the standard notation, respectively, for the sample space, the set of events, and the probability measure on them~\cite{Durrett10}.
Each element $\mathcal{F}_k$ of the filtration  $\mathcal{F}$ denotes the $\sigma$-algebra generated by the sequence $\{w_{t}\}_{t=0}^k$, while we set $\mathcal{F}_{-1} = \{\emptyset,\Omega\}$ as being the trivial $\sigma$-algebra.
%
We assume that the input $u_k$ is
a function of the system states $\{x_{t}\}_{t=0}^k$ and both $x_k$ and $u_k$
are $\mathcal{F}_{k-1}$-measurable random variables~\cite{Durrett10}. We also denote as $z_k = (x_k,u_k,w_{k})$ the vector of all the system variables at time $k$. 
Finally, we abbreviate as $\boldsymbol{z} = z_0, z_1, \ldots$ a system \emph{behavior} and as $\boldsymbol{z}^H = z_0, \ldots, z_{H-1}$ its truncation over the horizon $H$. 

\textbf{StSTL Syntax and Semantics.} StSTL formulas are defined over atomic predicates represented by \emph{chance constraints} of the form
\begin{equation}\label{eq:atomic_prop}
\mu ^{[p]} := \mathcal{P}\{\mu(v) \le 0\} \ge p,
\end{equation}
where $\mu(\cdot)$ is a real-valued measurable function, $v$ is a random
variable on the probability space $(\Omega, \mathcal{F}, \mathcal{P})$, and
$p \in [0,1]$. The truth value of $\mu ^{[p]}$ is interpreted based on the satisfaction of the chance constraint, i.e.,
$\mu ^{[p]}$ is true (denoted with $\top$) if and only if $\mu(v) \le 0$ holds
with probability larger than or equal to $p$. StSTL also supports deterministic predicates as a particular case. If $\mu(v)$ is deterministic, then $\mu^{[p]}$ holds for any value of $p$ if and only if $\mu(v) \le 0$ holds. In this case, we can omit the superscript $[p]$.
We define the syntax of an StSTL formula as follows:
\begin{equation} \label{eq:formula_form}
\psi := \mu ^{[p]} \;|\; \neg \psi \;|\; \psi \vee \phi \;|\; \psi\ \U_{[t_1,t_2]} \phi
\;|\; \G_{[t_1,t_2]} \psi,
\end{equation}
where $\mu ^{[p]}$ is an atomic predicate, $\psi$ and $\phi$ are StSTL formulas, $t_1, t_2 \in \mathbb{R}_+ \cup \{+\infty\}$, and $\U$ and $\G$
are, respectively, the  \emph{until} and \emph{globally} temporal operators. Other operators, such as \emph{conjunction} ($\land$), \emph{weak until} ($\W$), or \emph{eventually} ($\F$) are also supported and can be expressed using the operators in~\eqref{eq:formula_form}.

The semantics of an StSTL formula can be defined recursively as follows:
\begin{equation*}\label{eq:semantics}
{\small
\begin{aligned}
(\boldsymbol{z},k) &\models \mu ^{[p]} &\leftrightarrow\;&
\mathcal{P}\{\mu(z_k) \le 0\} \ge p, \\
(\boldsymbol{z},k) &\models \neg \psi &\leftrightarrow\;&
\neg( (\boldsymbol{z},k) \models \psi) \\
(\boldsymbol{z},k) &\models \psi \vee \phi &\leftrightarrow\;&
(\boldsymbol{z},k) \models \psi \vee (\boldsymbol{z},k) \models \phi, \\
(\boldsymbol{z},k) &\models \psi \U_{[t_1,t_2]} \phi &\leftrightarrow\;&
\exists i \in [k+t_1,k+t_2]: (\boldsymbol{z},i) \models \phi \land \\
& & & (\forall j \in {[k+t_1,i-1]}: (\boldsymbol{z},j) \models \psi),\\
(\boldsymbol{z},k) &\models \G_{[t_1,t_2]} \psi &\leftrightarrow\;&
\forall i \in [k+t_1,k+t_2]: (\boldsymbol{z},i) \models \psi.
\end{aligned}
}
\end{equation*}

As an example, $(\boldsymbol{z},k) \models \G_{[t_1,t_2]} \phi$ means that
$\phi$ holds for all times $t$ between $t_1$ and $t_2$.
%
Intervals may also be open or unbounded, e.g., of the form $[t_1,+\infty)$. 
In this paper, we focus on \emph{bounded} StSTL formulas, that is, formulas that contain no unbounded operators. 
StSTL reduces to STL for deterministic systems, with the exception that the atomic predicate has the form $\mu(v)\le 0$ rather than $\mu(v) > 0$, as in STL.
A difference between StSTL and PrSTL is in the interpretation of the negation of an atomic predicate. In PrSTL the semantics of negation is \emph{probabilistic}, i.e., if $(\boldsymbol{z},t) \models \lambda^{\epsilon_t}_{\alpha_t}$ holds for an atomic PrSTL predicate  $\lambda^{\epsilon_t}_{\alpha_t}$, which is equivalent to stating that $\mathcal{P}\{\lambda_{\alpha_t} (z_t)<0\} > 1 - \epsilon_t$, then $(\boldsymbol{z},t) \models \tilde{\neg} \lambda^{\epsilon_t}_{\alpha_t}$ is interpreted as
$\mathcal{P}\{\lambda_{\alpha_t} (z_t)>0\} > 1 - \epsilon_t$, so that $\tilde{\neg} \lambda^{\epsilon_t}_{\alpha_t}$
and $\lambda^{\epsilon_t}_{\alpha_t}$ can be true at the same time. 
StSTL keeps, instead, the standard semantics of \emph{logic negation}. 

\section{Problem Formulation}


We can concretely express the sets of behaviors $A$
and $G$ in a contract using temporal logic formulas~\cite{Nuzzo15b} and, in particular, StSTL formulas. We then define an StSTL A/G contract as a triple $(V,\phi_A,\phi_G)$, where
$\phi_A$ and $\phi_G$ are StSTL formulas over the set of variables $V$. The canonical form of $(V,\phi_A,\phi_G)$ can be achieved by setting $\phi_G := \phi_A \to \phi_G$.
The main contract operators can then be mapped into entailment of StSTL formulas. 
We define below the verification and synthesis problems addressed in this paper. 


\begin{problem}[Contract Consistency and Compatibility Checking] \label{prob:1}
Given a stochastic system representation $\pfr{S}$ as in~\eqref{eq:sys} and a bounded StSTL contract $C = (V, \phi_A, \phi_G)$ on the system variables $V$, determine whether $C$ is consistent (compatible), that is, whether $\phi_G$ ($\phi_A$) is satisfiable. 
\end{problem}

\begin{problem}[Contract Refinement Checking] \label{prob:2}
Given a stochastic system representation $\pfr{S}$ as in~\eqref{eq:sys} and bounded StSTL contracts $C_1 = (V, \phi_{A1}, \phi_{G1})$  and $C_2 =
(V, \phi_{A2}, \phi_{A2})$ on the system variables $V$, determine whether $C_1 \preceq C_2$, that is, 
$\phi_{A2} \rightarrow \phi_{A1}$ and $\phi_{A1} \rightarrow
\phi_{G2}$ are both valid.
\end{problem}

\begin{problem}[Synthesis from Contract] \label{prob:3}
Given a stochastic system representation $\pfr{S}$ as in~\eqref{eq:sys}, a bounded StSTL contract $C = (V, \phi_A, \phi_G)$  on the system variables $V$, and time horizon $H$, determine a control trajectory $\boldsymbol{u}^H $ such that $(\boldsymbol{z}^H,0) \models \phi_A  \to \phi_G$. 
\end{problem}

\begin{example}\label{sec:motiv_exmp}
We consider the following system description:
\begin{equation} \label{eq:motivdyn}
\begin{split}
x_{k+1} = \begin{bmatrix} 1 & 1\\ 0 & 1 \end{bmatrix} x_{k} + \begin{bmatrix} 1 + 0.3w_{k,1} & -0.2w_{k,2} \\ -0.2w_{k,2} & 1 + 0.3w_{k,1}\end{bmatrix} u_k,
\end{split}
\end{equation}
where $w_k = [w_{k,1},w_{k,2}]^T$  follows a standard Gaussian distribution, i.e., $w_k \sim \mathcal{N}(0,I)$ for all $k$, $I$ being the identity matrix. We assume that the first state variable at time $0$, $[1,0] x_{0}$, is in the interval $[1,2]$ and require that with probability smaller than $0.7$ the first state variable at time $2$ does not exceed $1$. We can formalize this requirement with the following StSTL contract $C_1 = (\phi_{A1}, \phi_{G1})$ in canonical form: 
\begin{equation} \label{eq:motivcon}
\begin{split}
\phi_{A1} & :=  (1 \leq [1,0] x_0) \land ([1,0] x_0 \leq 2), \\
\phi_{G1} & := \phi_{A1} \rightarrow \neg (\mathcal{P}\{[1,0]x_{2} \le 1\} \ge 0.7),
\end{split}
\end{equation}
where, 
for brevity, we drop the set of variables in the contract tuple. 
Assumptions and guarantees are expressed by logical combinations of arithmetic constraints over real numbers and chance constraints, all supported by StSTL. 
We intend to verify the \emph{consistency} of $C_1$.

Given the assumption on the distribution of $w_k$, it is possible to show that there exists a constant matrix $ \Lambda_1^{1/2} \in \mathbb{R}^{3\times 3}$ such that the constraint $\mathcal{P}\left\{[1,0] x_{2} \le 1\right\} \ge 0.7$ translates into a deterministic constraint\footnote{Details on how to compute such a matrix $\Lambda_1^{1/2}$ are provided in Sec.~\ref{sec:encoding}.} 
$f(x_0,u_0,u_1) \leq 0$, where 
%
\begin{align}\label{eq:exmp_chance_cons_equi} 
f(.)  = & [1,2] x_0 + [1,1,1,0]\begin{bmatrix} u_0 \\ u_1\end{bmatrix} -1 + \\
& + F^{-1}(0.7) \left\| \Lambda_1^{1/2}\begin{bmatrix} u_0 \\ u_{1}  \notag \\ 1
\end{bmatrix}\right\|_2,
\end{align}
$F^{-1}$ is the inverse cumulative distribution of a standard normal random variable, and $\left\| . \right\|_2$ is the $\ell_2$ norm. Hence, the contract is consistent if and only if there exists $(x_0,u_0,u_1)$ that satisfies
\begin{equation}\label{eq:exmp_consis2}
([1,0]x_0 < 1) \vee ([1,0]x_0 > 2) \vee f(x_0,u_0,u_1) > 0.
\end{equation}		
To solve this problem, we can translate~\eqref{eq:exmp_consis2} into a mixed integer program by applying encoding techniques proposed in the literature~\cite{raman2014model}. However, since one of the constraints in~\eqref{eq:exmp_consis2} is non-convex, using a nonlinear solver may be inefficient and usually requires the knowledge of bounding boxes for all the decision variables. Moreover, analytical expressions of chance constraints may not be even available in general~\cite{nemirovski2006convex}. 
Similar considerations hold for the problems of checking compatibility,  refinement, and for the generation of MPC schemes. 
\end{example}

Sec.~\ref{sec:encoding} addresses the issue highlighted in Example~\ref{sec:motiv_exmp} by providing techniques for systematically computing mixed integer linear  approximations of chance constraints and bounded StSTL formulas for three common classes of stochastic linear systems. To effectively perform the verification and synthesis tasks in Problem~\ref{prob:1}-\ref{prob:3}, we look for  both under- and over-approximations of StSTL formulas. For example, if the under-approximation of~\eqref{eq:exmp_consis2} is feasible, then we can conclude that  $C_1$ is consistent. However, infeasibility of the under-approximation is not sufficient to conclude about contract inconsistency; for this purpose, we need to prove that the over-approximation of~\eqref{eq:exmp_consis2} is infeasible. 

\ignore{
\secref{encoding}  For instance, sufficient and necessary conditions for the satisfiability of $\mathcal{P}\left((1,0)x_{2} \le 1\right) \ge 0.7$ in~\eqref{eq:exmp_chance_cons_equi} can be, respectively, expressed by the following linear constraints:

{\small
\begin{equation}\label{eq:suffi_neces_exmp}
\begin{split}
& (1,2)x_0 + [1,1,1,0]\begin{bmatrix} u_0 \\ u_1\end{bmatrix} - 1 + F^{-1}(0.7) \sum_{j=1}^5 \left|e_j^T T \begin{bmatrix} u_0 \\ u_1 \\ 1\end{bmatrix}\right| \le 0, \\
& (1,2)x_0 + [1,1,1,0]\begin{bmatrix} u_0 \\ u_1\end{bmatrix} - 1 + \frac{F^{-1}(0.7)}{\sqrt{5}} \sum_{j=1}^5 \left|e_j^T T \begin{bmatrix} u_0 \\ u_1 \\ 1\end{bmatrix}\right| \le 0,
\end{split}
\end{equation}
}
\noindent where $e_j^T$ is the $j$th row of the identity matrix $I$. The constraints in~\eqref{eq:suffi_neces_exmp} can be easily linearized and are, therefore, more tractable than the one in~\eqref{eq:exmp_chance_cons_equi}. 
}

\section{MIP Encoding of Bounded StSTL}\label{sec:encoding}

We present algorithms for the translation of bounded StSTL formulas into mixed integer constraints on the variables of a stochastic system.  
A MIP \emph{under-approximation} of an StSTL formula $\psi$ is a set of mixed integer constraints   $\mathcal{C}^S(\psi)$ whose feasibility is sufficient to ensure the satisfiability of $\psi$.
A MIP \emph{over-approximation} of $\psi$ is a set of mixed integer constraints $\mathcal{C}^N(\psi)$ which must be feasible if $\psi$ is satisfiable. 
When tractable closed-form translations of chance constraints are available, the formula under- and over-approximations coincide and provide an \emph{equivalent} encoding of the satisfiability problem. Otherwise, our framework provides under- and over-approximations in the form of mixed integer linear constraints. 
We start by discussing the translation of atomic predicates.


\subsection{MIP Translation of Chance Constraints}
\label{sec:handlingCons}

Our goal is to translate chance constraints into sets of
deterministic constraints that can be efficiently solved and provide a sound
formulation for our verification and synthesis tasks. 
%
%
Since approximation techniques  depend on the structure of the function
$\mu(\cdot)$ and the distribution of $z_k$ at each time $k$,
we detail solutions for three classes of dynamical systems and chance
constraints that arise in various application domains.
We denote by $S(\mu^{[p]}) \le 0$ the \emph{under-approximation} of the chance constraint, i.e., the set of mixed integer constraints
whose feasibility is sufficient to guarantee the predicate satisfaction. 
Similarly, we denote by $N(\mu^{[p]})\le 0$ the chance constraint \emph{over-approximation}, i.e., the set of constraints whose feasibility is necessary for the predicate satisfiability.

For simplicity, we present approximations of nonlinear constraints consisting of single linear constraints. Piecewise-affine approximations can also be used to arbitrarily improve the approximation accuracy~\cite{bradley1977applied} at higher computation costs. 


\subsubsection{Linear Systems with Additive and Control-Dependent Noise}\label{sec:class1}

We consider the class of stochastic linear systems governed by the following dynamics
\begin{equation} \label{eq:sys1}
\begin{split}
x_{k+1} &= A x_k + B_k u_k + \zeta_k, \\
[B_k,\zeta_k] &= [\bar B_k,\bar\zeta_k] + \sum_{l=1}^N [\tilde B_l,\tilde \zeta_l] w_{k,l},
\end{split}
\end{equation}
where $w_k = [w_{k,1},\ldots,w_{k,N}]^T \in \mathbb{R}^N$ follows the normal distribution $\mathcal{N}(\bar w_k, \Theta_{k})$, and $\bar B_k$ and $\bar\zeta_k$, for each $k$, and $\tilde B_l$ and $\tilde \zeta_l$, for each $l \in \{1,\ldots,N\}$, are constant matrices and vectors, respectively. The resulting matrix $B_k$ and vector $\zeta_k$ are stochastic and model, respectively, a multiplicative and and additive noise term.
This model has been used, for instance, to represent motion dynamics under corrupted control signals~\cite{harris1998signal} or 
networked control systems affected by channel fading~\cite{elia2005remote}. 
Requirements such as policy gains or bounds on the states 
for these systems are often expressed by the following chance constraint:
\begin{equation}\label{eq:type1_chance_cons}
\mathcal{P}\{\mu(z_k) \le 0\} \ge p, \;
\mu(z_k) = a^T x_k + b^T u_k + c.
\end{equation}

The next result provides an exact encoding for~\eqref{eq:type1_chance_cons}. Let $\boldsymbol{u}_{[0,k]}=\left[ u_0^T,\ldots,u_{k}^T \right]^T$ be the vector of the control inputs from $u_0$ to $u_{k}$. 
We denote by $\Theta_{k}^{(l_1 l_2)}$ the $l_1$-th row and $l_2$-th column element of the covariance matrix $\Theta_{k}$, and by $F^{-1}$ the inverse cumulative distribution function of a standard normal random variable.

\begin{thm}
The chance constraint \eqref{eq:type1_chance_cons} on the behaviors of the system in~\eqref{eq:sys1} is equivalent to
\begin{equation}\label{eq:linear_chance_cons_deter}
\lambda_1 (x_0, \boldsymbol{u}_{[0,k]}) + F^{-1}(p) \lambda_2 (x_0, \boldsymbol{u}_{[0,k]}) \le 0,
\end{equation}
where $\lambda_1$ is given by
\begin{equation}\label{eq:Lambda1}
\begin{split}
	\lambda_1(x_0, \boldsymbol{u}_{[0,k]}) & =   a^T A^k x_0 + b^T u_k + c \\
	&+ \sum_{t=1}^{k} a^T A^{k-t} (\bar\zeta_{t-1} + \bar B_{t-1} u_{t-1}) \\
	&+ \sum_{t=1}^{k} \sum_{l=1}^N a^T A^{k-t}
	(\tilde\zeta_{l}  + \tilde B_{l} u_{t-1}) \bar w_{t-1,l},
\end{split}
\end{equation}
and $\lambda_2$ is an $\ell_2$-norm of the system inputs
\begin{equation}\label{eq:Lambda2}
	\lambda_2(x_0, \boldsymbol{u}_{[0,k]}) = \left\|\Lambda_{k-1}^{1/2} \left[ \boldsymbol{u}_{[0,k-1]}^T, 1 \right]^T \right\|_2.
\end{equation}
The scaling matrix $\Lambda_{k-1}^{1/2}$ is deterministic for the given dynamics~\eqref{eq:sys1} and chance constraint~\eqref{eq:type1_chance_cons} and can be computed as a square root matrix of $\Lambda_{k-1}$, obtained as follows:
\begin{equation}\label{eq:Lambda_k_minus_1}
	\begin{split}
		\Lambda_{k-1} & = \begin{bmatrix} \Lambda_{1,1} & \Lambda_{1,2} \\ \Lambda_{1,2}^T & \Lambda_{2,2} \end{bmatrix}, \\
		\Lambda_{1,1} & = \mathrm{diag}(\alpha_{k-1}, \ldots, \alpha_0), \quad \Lambda_{1,2}  = [\beta_{k-1}, \ldots, \beta_0]^T,\\
		\Lambda_{2,2} & = \sum_{t = 1}^k \sum_{l_1 = 1}^N \sum_{l_2 = 1}^N a^T A^{k-t} \tilde{\zeta}_{l_1} a^T A^{k-t} \tilde{\zeta}_{l_2} \Theta_{t-1}^{(l_1 l_2)}, \\
		& \forall t\in \{0,\ldots,k-1\}: \\		 
		\alpha_t & = \sum_{l_1 = 1}^N \sum_{l_2 = 1}^N \tilde{B}_{l_1}^T (A^t)^T a a^T A^t \tilde{B}_{l_2} \Theta_{k-1-t}^{(l_1 l_2)}, \\
		\beta_t &= \sum_{l_1 = 1}^N \sum_{l_2 = 1}^N a^T A^t \tilde{\zeta}_{l_1} a^T A^t \tilde{B}_{l_2} \Theta_{k-1-t}^{(l_1 l_2)}. \\
	\end{split}
\end{equation}
\end{thm}

\begin{proof}
	The state $x_k$ of the stochastic system \eqref{eq:sys1} is known to be a linear function of the Gaussian sequence $\{w_t\}_{t=0}^{k-1}$, hence it follows a Gaussian distribution. This also applies to $\mu(z_k)$. In fact, by substituting \eqref{eq:sys1} into the expression for $\mu(z_k)$, we obtain
	\begin{equation}\label{eq:sys_random_linear_cons_linear}
	\begin{split}
	\mu(z_k) = {} &  a^T A^k x_0 + b^T u_k + c \\
	&+ \sum_{t=1}^{k} a^T A^{k-t} (\bar\zeta_{t-1} + \bar B_{t-1} u_{t-1}) \\
	&+ \sum_{t=1}^{k} \sum_{l=1}^N a^T A^{k-t} (\tilde\zeta_{l}  + \tilde B_{l} u_{t-1}) w_{t-1,l}.
	\end{split}
	\end{equation}
Therefore, $\mu(z_k)$ is linear in the random variables $w_{t-1,l}$, $l\in\{1,\ldots,N\}$ and also follows a Gaussian distribution. Next, we derive the mean and the standard deviation of $\mu(z_k)$.
	
	Since the random vector $w_{t-1}$ follows the Gaussian distribution $\mathcal{N}(\bar w_{t-1}, \Theta_{k})$, the expectation of its $l$-th element $w_{t-1,l}$ is $\bar w_{t-1,l}$. Let  $\lambda_1 = \mathbb{E} \{\mu(z_k)\}$ be the expectation of $\mu(z_k)$. Then, we obtain
	\begin{align*}
	\lambda_1 & = a^T A^k x_0 + b^T u_k + c + \sum_{t=1}^{k} a^T A^{k-t} (\bar\zeta_{t-1} + \bar B_{t-1} u_{t-1}) \\
	&+ \sum_{t=1}^{k} \sum_{l=1}^N a^T A^{k-t}
	(\tilde\zeta_{l}  + \tilde B_{l} u_{t-1}) \bar w_{t-1,l},
	\end{align*}
	which is \eqref{eq:Lambda1}.  To derive the standard deviation of $\mu(z_k)$, we first write $\tilde{\mu} = \mu(z_k) - \mathbb{E} \{\mu(z_k)\}$ into a more compact form,
	\begin{equation*}
	\tilde{\mu} = \mathcal{B}_{k-1} \boldsymbol{u}_{[0,k-1]} + \mathcal{Z}_{k-1} = \left[\mathcal{B}_{k-1}, \; \mathcal{Z}_{k-1}\right] \begin{bmatrix} \boldsymbol{u}_{[0,k-1]} \\ 1 \end{bmatrix},
	\end{equation*}
	where $\mathcal{B}_{k-1}$ and $\mathcal{Z}_{k-1}$ are random matrices defined as follows
	\begin{equation*}
	\begin{split}	
	\mathcal{B}_{k-1} & = \sum_{l=1}^N \left[ a^T A^{k-1} \tilde B_{l} \tilde w_{0,l},\; \ldots, \; a^T \tilde B_{l} \tilde w_{k-1,l} \right], \\
	\mathcal{Z}_{k-1} & = \sum_{t=1}^k \sum_{l=1}^N a^T A^{k-t} \tilde\zeta_{l} \tilde w_{t-1,l}, \\
	\tilde w_{t-1,l} & = w_{t-1,l} - \bar w_{t-1,l}.
	\end{split}
	\end{equation*}	
	Then, we obtain
	\begin{equation*}
	\begin{split}
	 \mathbb{E}\{\tilde{\mu}^2\} & = \mathbb{E} \left\{ \left[ \boldsymbol{u}_{[0,k-1]}^T, 1 \right] \begin{bmatrix}\mathcal{B}_{k-1}^T \\ \mathcal{Z}_{k-1}^T\end{bmatrix} \left[\mathcal{B}_{k-1}, \; \mathcal{Z}_{k-1}\right] \begin{bmatrix} \boldsymbol{u}_{[0,k-1]} \\ 1 \end{bmatrix} \right\} \\
	& =  \left[ \boldsymbol{u}_{[0,k-1]}^T, 1 \right] \mathbb{E}\left\{ \begin{bmatrix}\mathcal{B}_{k-1}^T \\ \mathcal{Z}_{k-1}^T\end{bmatrix} \left[\mathcal{B}_{k-1}, \; \mathcal{Z}_{k-1}\right] \right\} \begin{bmatrix} \boldsymbol{u}_{[0,k-1]} \\ 1 \end{bmatrix}
	\end{split}
	\end{equation*}
	and, by renaming the positive semidefinite matrix
	\begin{equation}\label{eq:Lambda}
		\Lambda_{k-1} = \mathbb{E}\left\{ \begin{bmatrix}\mathcal{B}_{k-1}^T \\ \mathcal{Z}_{k-1}^T\end{bmatrix} \left[\mathcal{B}_{k-1}, \; \mathcal{Z}_{k-1}\right] \right\},
	\end{equation}
	we can finally write
	\[
		\mathbb{E}\{\tilde{\mu}^2\} = \left\|\Lambda_{k-1}^{1/2} \left[ \boldsymbol{u}_{[0,k-1]}^T, 1 \right]^T \right\|_2^2 = \lambda^2_2,
	\]
	saying that  $\lambda_2$ in \eqref{eq:Lambda2} corresponds to the standard deviation of $\mu(z_k)$.  
	The full expression for $\Lambda_{k-1}$ in \eqref{eq:Lambda} can be obtained by computing the expectation $\mathbb{E}\{\cdot\}$ and observing that $\mathbb{E}\{\tilde w_{t,l}\} = 0$ and $\mathbb{E}\{\tilde w_{t,l_1} \tilde w_{t,l_2}\} = \Theta_t^{(l_1 l_2)}$, which leads to \eqref{eq:Lambda_k_minus_1}.
	
	Finally, the chance constraint \eqref{eq:type1_chance_cons} on the random variable $\mu(z_k)$ following the distribution $\mathcal{N}(\lambda_1, \lambda_2)$ is equivalent to
	\[
	\lambda_1 + F^{-1}(p) \lambda_2 \le 0,
	\]
	which corresponds to \eqref{eq:linear_chance_cons_deter}, as we wanted to prove.
\end{proof}

In~\eqref{eq:linear_chance_cons_deter}, $\lambda_1$ is a linear function of its variables, and $\lambda_2$ is an $\ell_2$-norm of the system inputs. While~\eqref{eq:linear_chance_cons_deter} is convex when $p \ge 0.5$, this is no longer the case for  $p < 0.5$.
In both cases, we provide an efficient linear approximation by applying a classical norm inequality
to derive lower and upper bound functions $\lambda_2^u$ and $\lambda_2^l$  for $\lambda_2(.)$  as follows:
\begin{equation*}
\begin{split}
\lambda_2^u (x_0, \boldsymbol{u}_{[0,k]}) &=
\sum_{j=1}^{k n_u + 1} \left|e_j^T \Lambda_{k-1}^{1/2} \begin{bmatrix} \boldsymbol{u}_{[0,k-1]} \\ 1 \end{bmatrix}\right|, \\
\lambda_2^l (x_0, \boldsymbol{u}_{[0,k]}) &= \frac{1}{\sqrt{k n_u + 1}}\lambda_2^u
(x_0, \boldsymbol{u}_{[0,k]}),
\end{split}
\end{equation*}
where $e_j^T$ is the $j$-th row of the identity matrix $I$ and $n_u$ is the dimension of $u_k$.
Then, an under-approximation $S(\mu^{[p]}) \le 0$ for~\eqref{eq:linear_chance_cons_deter} is given by
{\small
\begin{equation}\label{eq:linear_chance_cons_deter_suffi}
\begin{cases}
\lambda_1 (x_0, \boldsymbol{u}_{[0,k]}) + F^{-1}(p) \lambda_2^u (x_0, \boldsymbol{u}_{[0,k]}) \le 0, \quad p \ge 0.5 \\
\lambda_1 (x_0, \boldsymbol{u}_{[0,k]}) + F^{-1}(p) \lambda_2^l (x_0, \boldsymbol{u}_{[0,k]}) \le 0, \quad p < 0.5.
\end{cases}
\end{equation}
}
Similarly, an over-approximation $N(\mu^{[p]}) \le 0$ can be obtained as follows: 
{\small
\begin{equation}\label{eq:linear_chance_cons_deter_neces}
\begin{cases}
\lambda_1 (x_0, \boldsymbol{u}_{[0,k]}) + F^{-1}(p) \lambda_2^l (x_0, \boldsymbol{u}_{[0,k]}) \le 0, \quad p \ge 0.5 \\
\lambda_1 (x_0, \boldsymbol{u}_{[0,k]}) + F^{-1}(p) \lambda_2^u (x_0, \boldsymbol{u}_{[0,k]}) \le 0, \quad p < 0.5.
\end{cases}
\end{equation}
}

\begin{table*}[t]
\centering
\caption{Deterministic encodings of the chance constraint $\mathcal{P}\{\mu(z_k)\le 0\}\ge p$}
\label{tab:chance_cons_formu}
\begin{tabular}{c|c|c|c|c|c}
\hline
	System dynamics
	& Constraint function $\mu(z_k)$
	& Distribution of $w_k$
	& Exact 
	& \begin{tabular}{@{}c@{}} Under-approx \\ $S(\mu^{[p]})(z_k)\le 0$ \end{tabular}
	& \begin{tabular}{@{}c@{}} Over-approx \\ $N(\mu^{[p]})(z_k)\le 0$ \end{tabular} \\
\hline
	$\begin{aligned}
		x_{k+1} &= A x_k + B_k u_k + \zeta_k, \\
		[B_k,\zeta_k] &= [\bar B_k,\bar\zeta_k] + \textstyle\sum_{l=1}^H [\tilde B_l,\tilde \zeta_l] w_{k,l}
	\end{aligned}$
	& $a^T x_k + b^T u_k + c$
	& Normal $\mathcal{N}(\bar w_k, \Theta_{k})$ 
	& \eqref{eq:linear_chance_cons_deter}
	& \eqref{eq:linear_chance_cons_deter_suffi} 
	& \eqref{eq:linear_chance_cons_deter_neces} \\
\hline
	$\begin{aligned}
		x_{k+1} &= A_k x_k + B_k u_k + \zeta_k, \\
		[A_k,B_k,\zeta_k] &= [A(w_k),B(w_k), \zeta(w_k)]
	\end{aligned}$
	& $a^T x_k + b^T u_k + c$
	& \begin{tabular}{@{}c@{}}
		Discrete-time finite-state \\
	    Markov chain
	\end{tabular}
	& \eqref{eq:chance_cons_Markov_jump_equiv}
	& \eqref{eq:chance_cons_Markov_jump_equiv} 
	& \eqref{eq:chance_cons_Markov_jump_equiv} \\
\hline
	$\begin{aligned}
		x_{k+1} &= A x_k + B u_k, \\
		\xi_k &= \left[ x_k^T ,\; u_k^T \right]^T
	\end{aligned}$
	& $w_{k}^T \xi_k + c$
	& Normal $\mathcal{N}(\bar w_k, \Theta_{k})$ 
	& \eqref{eq:deter_sys_Gaus_cons_equi} 
	& \eqref{eq:deter_sys_Gaus_cons_suffi}
	& \eqref{eq:deter_sys_Gaus_cons_neces} \\
\hline
\end{tabular}
\end{table*}

\subsubsection{Markovian Jump Linear Systems}
\label{sec:class2}

Markovian jump linear systems are frequently used to model discrete transitions, for instance, due to component failures, abrupt disturbances, or changes in the operating
points of linearized models of nonlinear systems~\cite{de2006mode}. They are characterized by the following dynamics
\begin{equation}\label{eq:sys2}
\begin{split}
x_{k+1} &= A_k x_k + B_k u_k + \zeta_k, \\
[A_k,B_k,\zeta_k] &= [A(w_k),B(w_k),\zeta(w_k)],
\end{split}
\end{equation}
where $A_k, B_k, \zeta_k$ are all functions of $w_k$,  and the sequence
$\{w_k\}_{k=0}^\infty$ is a discrete-time finite-state Markov chain. We assume that, 
for all $k$, $w_k$ takes a value
$w^{l_k}\in\{w^{0}, \ldots, w^{N}\}$. 

We use $\boldsymbol{w}_{[0,k-1]}$ and $\boldsymbol{w}^{[l_0,l_{k-1}]}$ to denote, respectively, the random trajectory $w_0,\ldots,w_{k-1}$ and a particular scenario $w^{l_0},\ldots,w^{l_{k-1}}$. $\mathcal{P}\{\boldsymbol{w}_{[0,k-1]} = \boldsymbol{w}^{[l_0,l_{k-1}]}\}$ is the probability of occurrence of the scenario $\boldsymbol{w}^{[l_0,l_{k-1}]}$. Moreover, for each scenario, we introduce a binary variable $b(\boldsymbol{w}^{[l_0,l_{k-1}]})$ which evaluates to $1$ if and only if $\mu(z_k) \le 0$ holds for the scenario $\boldsymbol{w}^{[l_0,l_{k-1}]}$. Then, an exact encoding for the chance constraint~\eqref{eq:type1_chance_cons} on a Markovian jump linear system is given by the following result.

\begin{thm}
The chance constraint~\eqref{eq:type1_chance_cons} on the behaviors of the system in~\eqref{eq:sys2} is equivalent to the following MIL constraints
\begin{equation}\label{eq:chance_cons_Markov_jump_equiv}
{\small
\begin{cases}
\sum\limits_{t = 0}^{k-1} \sum\limits_{l_{t} = 0}^N b(\boldsymbol{w}^{[l_0,l_{k-1}]})
\mathcal{P}\{\boldsymbol{w}_{[0,k-1]} = \boldsymbol{w}^{[l_0,l_{k-1}]}\} \ge p, \\
\lambda(x_0, \boldsymbol{u}_{[0,k]}, \boldsymbol{w}^{[l_0,l_{k-1}]}) \le 0 \leftrightarrow b(\boldsymbol{w}^{[l_0,l_{k-1}]}) = 1, \\
\end{cases}
}
\end{equation}
where $\lambda(x_0, \boldsymbol{u}_{[0,k]}, \boldsymbol{w}^{[l_0,l_{k-1}]}) \le 0$ enforces that the particular scenario satisfies the chance constraint. $\lambda(\cdot)$
can be computed as follows:
\begin{equation}\label{eq:markov_jump_lambda}
{\small
\begin{split}
\lambda(\cdot) ={} & a^T \mathcal{A}_{k-1} x_0 + \mathcal{B}_{k-1} \boldsymbol{u}_{[0,k]} + \mathcal{Z}_{k-1} + c  \\ 
\mathcal{A}_{k-1} ={} & \left[  A(w^{l_{k-1}}), \cdots, A(w^{l_0})\right], \\
\mathcal{B}_{k-1} ={} & \left[ a^T \mathcal{A}_{k-1} B(w^{l_0}),\; \ldots,\; a^T B(w^{l_{k-1}}), b^T \right] \\
\mathcal{Z}_{k-1} ={} & a^T \mathcal{A}_{k-1} \zeta(w^{l_0}) + \ldots + a^T \zeta(w^{l_{k-1}}), 
\end{split}
}
\end{equation}
with $\boldsymbol{u}_{[0,k]} = [u_{0}^T,\ldots,u_{k}^T]^T$.
\end{thm}

\begin{proof}
	For a given scenario $\boldsymbol{w}^{[l_0,l_{k-1}]}$ for the Markovian jump linear system in~\eqref{eq:sys2}, the system state $x_k$ is a deterministic function of $\boldsymbol{u}_{[0,k-1]} = [u_{0}^T,\ldots,u_{k-1}^T]^T$. We can then express the constraint $\mu(z) = a^T x_k + b_i^T u_k + c \le 0$ as in~\eqref{eq:markov_jump_lambda}. The probability $\mathcal{P}\{a^T x_k + b^T u_k + c \le 0\}$ can be computed by considering all the possible scenarios for $\boldsymbol{w}_{[0,k-1]}$ as follows:
	\begin{equation}\label{eq:chanc_cons_Markov_jump}
	\begin{split}
	& \mathcal{P}\{a^T x_k + b^T u_k + c \le 0\} \\
	& \;\; = \sum_{t = 0}^{k-1} \sum_{l_{t} = 0}^N
	\mathcal{P}\{a^T x_k + b^T u_k + c \le 0, \boldsymbol{w}^{[l_0,l_{k-1}]}\} \\
	& \;\; = \sum_{t = 0}^{k-1} \sum_{l_{t} = 0}^N
	\mathcal{P}\{a^T x_k + b^T u_k + c \le 0 | \boldsymbol{w}^{[l_0,l_{k-1}]}\}\cdot \\
	& \phantom{\;\; = \sum_{t = 0}^{k-1} \sum_{l_{t} = 1}^H}
	\; \mathcal{P}\{\boldsymbol{w}_{[0,k-1]} = \boldsymbol{w}^{[l_0,l_{k-1}]}\}.
	\end{split}
	\end{equation}
Whether the constraint $a^T x_k + b^T u_k + c \le 0$ is satisfied or not under a given scenario $\boldsymbol{w}^{[l_0,l_{k-1}]}$ is a deterministic event, hence the probability $\mathcal{P}\{a^T x_k + b^T u_k + c \le 0 | \boldsymbol{w}^{[l_0,l_{k-1}]}\}$ is either $1$ or $0$, and corresponds to the value of the binary indicator variable $b(\boldsymbol{w}^{[l_0,l_{k-1}]})$. By introducing  $b(\boldsymbol{w}^{[l_0,l_{k-1}]})$ into~\eqref{eq:chanc_cons_Markov_jump}, the chance constraint
	$\mathcal{P}\{a^T x_k + b^T u_k + c \le 0\} \ge p$ reduces to the first constraint in~\eqref{eq:chance_cons_Markov_jump_equiv}, where the probability
	$\mathcal{P}\{ \boldsymbol{w}_{[0,k-1]} = \boldsymbol{w}^{[l_0,l_{k-1}]}\}$
	is given by the transition probability matrix of the Markov chain. The second constraint in~\eqref{eq:chance_cons_Markov_jump_equiv} directly descends from the definition of $b(\boldsymbol{w}^{[l_0,l_{k-1}]})$.
	Therefore, constraints~\eqref{eq:chance_cons_Markov_jump_equiv} and~\eqref{eq:markov_jump_lambda} 
provide an exact encoding of the chance constraint~\eqref{eq:type1_chance_cons} for a Markovian jump linear system, which is what we wanted to prove. The implication in~\eqref{eq:chance_cons_Markov_jump_equiv} can be translated into MIL constraints using standard techniques~\cite{l2008operations}. 
%
\end{proof}

\subsubsection{Deterministic Systems with Measurement Noise}\label{sec:class3}

We consider a system
\begin{equation*}
x_{k+1} = A x_k + B u_k, \quad
\xi_k = \begin{bmatrix}
x_k \\
u_k
\end{bmatrix},
\end{equation*}
subject to constraints of the form
\begin{equation}\label{eq:chance_cons_deter}
\begin{split}
\mathcal{P}\{\mu(z_k)\le 0\} \ge p, \quad \mu(z_k) = w_{k}^T \xi_k + c,\;
\end{split}
\end{equation}
where $w_k$ follows the normal distribution $\mathcal{N}(\bar w_k, \Theta_{k})$. This setting can be used to represent uncertainties in perception, e.g., in the detection of environment obstacles to the trajectory of autonomous systems~\cite{sadigh2016}.
As for the system in Sec.~\ref{sec:class1}, an exact translation of~\eqref{eq:chance_cons_deter}~\cite{sadigh2016} leads to
\begin{equation}\label{eq:deter_sys_Gaus_cons_equi}
\begin{split}
\bar w_{k}^T \xi_k + c + F^{-1}(p) \left\|\Theta_k^{1/2} \xi_k\right\|_2 \le 0,
\end{split}
\end{equation}
which may result in non-convex constraint. Again, by using 
a norm inequality to bound the $\ell_2$-norm in~\eqref{eq:deter_sys_Gaus_cons_equi}, we provide an under-approximation of~\eqref{eq:chance_cons_deter} in the form
\begin{equation}\label{eq:deter_sys_Gaus_cons_suffi}
{\small
\begin{split}
\begin{cases}
\bar w_{k}^T \xi_k + c + F^{-1}(p) \sum\limits_{j=1}^{n_z} \left|e_j^T \Theta_k^{1/2} \xi_k\right|
\le 0, \;\; p \ge 0.5, \\
\bar w_{k}^T \xi_k + c + \frac{F^{-1}(p)}{\sqrt{n_\xi}} \sum\limits_{j=1}^{n_\xi}
\left|e_j^T \Theta_k^{1/2} \xi_k\right| \le 0, \;\; p < 0.5,
\end{cases} \\
\end{split}
}
\end{equation}
where $e_j$ is the $j$-th column of the identity matrix, and an over-approximation in the form
\begin{equation}\label{eq:deter_sys_Gaus_cons_neces}
{\small
\begin{cases}
\bar w_{k}^T \xi_k + c + \frac{F^{-1}(p)}{\sqrt{n_\xi}} \sum\limits_{j=1}^{n_\xi} \left|e_j^T \Theta_k^{1/2} \xi_k\right| \le 0, \;\; p \ge 0.5, \\
\bar w_{k}^T \xi_k + c + F^{-1}(p) \sum\limits_{j=1}^{n_\xi} \left|e_j^T \Theta_k^{1/2} \xi_k\right| \le 0, \;\; p < 0.5.
\end{cases}
}
\end{equation}

Table~\ref{tab:chance_cons_formu} provides a summary of the encodings in this section. 



\subsection{MIP Under-Approximation}\label{sec:suffi_encode}

We construct a MIP under-approximation $\mathcal{C}_k^S(\psi)$ of a formula $\psi$ by assigning a binary variable $b^{S}_{k}(\psi)$ to the formula such that
$b^{S}_{k} (\psi) = 1 \to (\boldsymbol{z},k) \models\psi$.
We then traverse the parse tree of  $\psi$ and associate binary
variables with all the sub-formulas in $\psi$.
Following the semantics in Sec.~\ref{eq:semantics}, the logical
relation between $\psi$ and its sub-formulas is then recursively captured using mixed integer constraints. The translation terminates when all the atomic
predicates are translated. 

Our encoding is different from the ones previously proposed for deterministic STL formulas~\cite{raman2014model}, in that the truth value of the Boolean variable $b$ associated to each atomic predicate $(\mu \leq 0)$ is not equivalent to the predicate satisfaction. Instead, $b = 1$ is only a sufficient condition for predicate satisfaction, as we are only able to associate $b$ with an under-approximation  $S(\mu^{[p]})(z_{k}) \le 0$. Because $b=0$ cannot encode the logical negation of the predicate, we deal with atomic predicates and their negations separately. Specifically, we  convert any formula into its negation normal form and associate distinct Boolean variables, e.g., $b$ and $\bar{b}$, to each atomic predicate and its negation, respectively. 
We use both $b$ and $\bar{b}$ to translate any Boolean and temporal operator involving the predicate or its negation in the formula. We illustrate this approach on some special cases below.

$\bm{\psi = \mu^{[p]}}$: We requires that $b_{k}^S(\mu^{[p]}) = 1$ implies the feasibility of a sufficient condition for $(\boldsymbol{z},k)\models \mu^{[p]}$ by the following constraint
\begin{equation}\label{eq:atomic_milp}
S(\mu^{[p]})(z_{k}) \le (1 - b_{k}^S(\mu^{[p]}))M,
\end{equation}
where $M$ is a sufficiently large positive constant (``big-$M$'' encoding technique)~\cite{l2008operations}, and $S(\mu^{[p]})(z_{k}) \le 0$ is the chance constraint under-approximation. 


$\bm{\psi = \neg\mu^{[p]}}$: 
If an under-approximation $S(\neg\mu^{[p]})(z_{k}) \le 0$ is available, then we require
\begin{equation}\label{eq:neg_atomic_milp}
\begin{split}
S(\neg\mu^{[p]})(z_{k}) \le (1 - b_{k}^S(\neg\mu^{[p]}))M.
\end{split}
\end{equation}
Otherwise, we recall that $\mathcal{P}(\mu(z_k) \leq 0) < p$
is equivalent to $\mathcal{P}(\mu(z_k) > 0) > 1-p$.
To bring this predicate into a standard form, we require that $\mathcal{P}(-\mu(z_k) + \epsilon\leq 0) \ge 1-p + \epsilon$, where $\epsilon > 0$ is a sufficiently small real constant. We can then use the encoding in~\eqref{eq:atomic_milp} to obtain
\begin{equation}\label{eq:neg_atomic_milp3}
S((-\mu + \epsilon)^{[1-p + \epsilon]})(z_k) \le (1 - b_{k}^S(\neg\mu^{[p]}))M.
\end{equation}

$\bm{\psi = \G_{[t_1,t_2]} \phi}$: To encode the bounded \emph{globally} predicate we add to $\mathcal{C}_k^S(\psi)$ the mixed integer linear constraint
\begin{equation}\label{eq:global_milp}
b_k^S(\G_{[t_1,t_2]} \phi) \leftrightarrow  \wedge_{i=t_1}^{t_2} b_{k+i}^S(\phi),
\end{equation}
requiring that $b_k^S(\G_{[t_1,t_2]} \phi) = 1$ if and only if $b_{k+i}^S(\phi) = 1$
for all $i \in [t_1, t_2]$.
The conjunction of the $b_{k+i}^S(\phi)$ is then translated into
mixed integer linear constraints using standard techniques~\cite{raman2014model}.

$\bm{\psi = \neg \G_{[t_1,t_2]} \phi}$: When \emph{globally} is negated, we augment $\mathcal{C}_k^S(\psi)$ with the mixed integer linear
constraint
\begin{equation}\label{eq:neg_global_milp}
b_k^S(\neg (\G_{[t_1,t_2]} \phi)) \leftrightarrow \vee_{i=t_1}^{t_2} b_{k+i}^S(\neg\phi),
\end{equation}
showing how we push the negation of a formula to its sub-formulas in a recursive fashion until we reach the atomic predicates.

For brevity, we omit the encoding for the other temporal operators, which directly follows from the semantics in Sec.~\ref{eq:semantics} and the approach in~\eqref{eq:global_milp} and~\eqref{eq:neg_global_milp}. 
%
%
%
If~\eqref{eq:atomic_milp} and~\eqref{eq:neg_atomic_milp3} are linear, then $\mathcal{C}_k^S(\psi)$ is a mixed integer linear constraint set.
Based on the above procedure, the following theorem summarizes the property of the MIP under-approximation.

\begin{thm}\label{thm:suffi_encoding}
	$\mathcal{C}_k^S(\psi)$ is a MIP under-approximation of $\psi$, i.e., if $\mathcal{C}_k^S(\psi)$ is feasible and $\boldsymbol{z}^*$ is a solution, then $\psi$ is satisfiable and $(\boldsymbol{z}^*, k)\models\psi$.
\end{thm}

\begin{proof}
	We first prove the theorem for the atomic predicates $\mu^{[p]}$ and $\neg\mu^{[p]}$. We observe that $\mathcal{C}_k^S(\mu^{[p]})$ is equivalent to the conjunction of the constraints $(b_k^S(\mu^{[p]}) = 1)$ and~\eqref{eq:atomic_milp}. 
	If $\mathcal{C}_k^S(\mu^{[p]})$ is feasible, then $S(\mu^{[p]})(z_k)\le 0$ must hold. 
	Since $S(\mu^{[p]})(z_k)\le 0$ is a sufficient condition for the satisfaction of the predicate, we conclude $(\boldsymbol{z}^*,k)\models \mu ^{[p]}$. Similarly, the feasibility of $\mathcal{C}_k^S(\neg\mu^{[p]})$ implies  $(\boldsymbol{z}^*,k)\models \neg \mu ^{[p]}$ using constraint~\eqref{eq:neg_atomic_milp}. 	
	
	We now consider a formula $\psi$ such that Theorem \ref{thm:suffi_encoding} holds for all its sub-formulas. Without loss of generality, we discuss $\psi = \phi_1 \U_{[t_1,t_2]}\phi_2$; the same proof structure can be applied to other temporal or logical operators. 
	$\mathcal{C}_k^S(\psi)$ contains the following constraints
	\begin{equation*}
	\begin{split}
	& b_k^S(\psi) = 1,\; b_k^S(\psi) = \vee_{i=t_1}^{t_2} (b_{k+i}^S(\phi_2)
	\wedge_{j=t_1}^{i-1}b_{k+j}^S(\phi_1)), \\
	& \mathcal{C}_{k+i}^S(\phi_1)\setminus \{b_{k+i}^S(\phi_1) = 1\},\;
	\mathcal{C}_{k+j}^S(\phi_2)\setminus \{b_{k+j}^S(\phi_2) = 1\},
	\end{split}
	\end{equation*}
	for all $i\in [t_1,t_2]$ and $j\in [t_1,t_2-1]$. We use $\mathcal{C}_{k+i}^S(\phi_1)\setminus \{b_{k+i}^S(\phi_1) = 1\}$ to denote the set of constraints in $\mathcal{C}_{k+i}^S(\phi_1)$ except for the constraint $(b_{k+i}^S(\phi_1) = 1)$. If $\mathcal{C}_k^S(\psi)$ is feasible, then $b_k^S(\psi) = 1$ must hold, hence there exists $i\in [t_1,t_2]$ such that
	$b_{k+i}^S(\phi_2) \wedge_{j=t_1}^{i-1}b_{k+j}^S(\phi_1) = 1$.
	We then obtain that $b_{k+i}^S(\phi_2) = 1$ holds as well as $b_{k+j}^S(\phi_1) = 1$, $\forall \ j \in [t_1,i-1]$. This ensures that $\mathcal{C}_{k+i}^S(\phi_1)$ and $\mathcal{C}_{k+j}^S(\phi_2)$, $\forall \ j \in [t_1,i-1]$, are feasible. 
	Since Theorem \ref{thm:suffi_encoding} holds for $\phi_1$ and $\phi_2$, we also have 
	$(\boldsymbol{z}^*,k+i)\models\phi_2$ and 
	$(\boldsymbol{z}^*,k+j)\models\phi_1$ $\forall \ j \in [t_1,i-1]$, hence 
	$(\boldsymbol{z}^*,k)\models\phi_1 \U_{[t_1,t_2]}\phi_2$, which is what we wanted to prove.
\end{proof}

It is possible that both the  $\mathcal{C}_k^S(\psi)$ and $\mathcal{C}_k^S(\neg\psi)$ 
under-approximations are infeasible, in which case 
we cannot make any conclusion on whether $\psi$ or $\neg\psi$ are satisfiable.
To conclude on the unsatisfiability of a formula, we resort to a MIP over-approximation.

\subsection{MIP Over-Approximation}\label{sec:neces_encoding}

To generate an over-approximation of $\psi$, we associate a binary variable $b^{N}_{k} (\psi)$ to $\psi$ and seek for a set of mixed integer constraints $\mathcal{C}_k^N(\psi)$ so that
$(\boldsymbol{z},k) \models\psi \rightarrow b^{N}_{k} (\psi) = 1$.
Creating an over-approximation only differs in
the interpretation of the atomic propositions, since we now use deterministic
mixed integer constraints that are necessary for the satisfaction of the chance constraints in the formula. As in Sec.~\ref{sec:suffi_encode}, we deal with an atomic predicate and its negation separately, and provide necessary condition for their satisfaction as follows.

$\bm{\psi = \mu^{[p]}}$: We assign a binary variable $b_{k}^N(\mu^{[p]})$ so that, if the over-approximation $N(\mu^{[p]})(z_k) \le
0$ is not satisfied, then $b_{k}^N(\mu^{[p]})$ is false. We, therefore, add the following mixed integer constraint: 
\begin{equation}\label{eq:atomic_milp_neces}
\begin{split}
N(\mu^{[p]})(z_k) \le (1 - b_k^N(\mu^{[p]}))M,
\end{split}
\end{equation}
%
where $M$ is a large enough positive constant~\cite{l2008operations}.

$\bm{\psi = \neg\mu^{[p]}}$: If an over-approximation $N(\neg\mu^{[p]})(z_k) \le 0$ is available, then we add a binary variable $b_{k}^N(\neg\mu^{[p]})$ and the mixed integer constraint
\begin{equation}\label{eq:neg_atomic_milp_neces}
\begin{split}
N(\neg\mu^{[p]})(z_k) \le (1 - b_k^N(\neg\mu^{[p]}))M.
\end{split}
\end{equation}
Otherwise, since $\mathcal{P}(\mu(z_k) \leq 0) < p$
implies $\mathcal{P}(-\mu(z_k) \leq 0) \ge 1- p$ we require
\begin{equation}\label{eq:neg_atomic_milp_neces2}
\begin{split}
N((-\mu)^{[1-p]})(z_k) \le (1 - b_k^N(\neg\mu^{[p]}))M.
\end{split}
\end{equation}



Other logic and temporal operators are encoded as in Sec.~\ref{sec:suffi_encode}. 
By similar arguments, we obtain the result below.
\begin{thm}\label{thm:neces_encoding}
	$\mathcal{C}_k^N(\psi)$ is a MIP over-approximation for the formula $\psi$, i.e., if $\mathcal{C}_k^N(\psi)$ is infeasible,  then  $\psi$ is unsatisfiable.
\end{thm}

\begin{proof}
	We  need to prove that $(\boldsymbol{z},k)\models\psi$ is sufficient for the feasibility of $\mathcal{C}_k^N(\psi)$. Let first $\psi$ be the atomic proposition $\mu^{[p]}$.
	Since $N(\mu^{[p]})(z_k) \le 0$ is a necessary condition for the satisfaction of $\mu^{[p]}$, we obtain
	$(\boldsymbol{z},k)\models \mu^{[p]} \to N(\mu^{[p]})(z_k) \le 0$.
	Then, if $\mu^{[p]}$ is satisfiable, the conjunction of \eqref{eq:atomic_milp_neces} and $b_k^N(\mu^{[p]}) = 1$ holds, which is equivalent to the feasibility of $\mathcal{C}_k^N(\psi)$. A similar argument can be used for $\neg\mu^{[p]}$.
	
	When $\psi$ is a generic formula, let Theorem \ref{thm:neces_encoding} hold for the sub-formulas of $\psi$. Then, if a sub-formula is satisfiable, its over-approximation is feasible. Without loss of generality, we consider $\psi = \neg(\phi_1 \U_{[t_1,t_2]} \phi_2)$. $(\boldsymbol{z},k)\models\psi$ is equivalent to 
	\[
	\wedge_{i=t_1}^{t_2} ((\boldsymbol{z},k+i)\models\neg\phi_2 \vee_{j=t_1}^{i-1}
	(\boldsymbol{z},k+j)\models\neg\phi_1)
	\]
	being true, meaning that for all $i \in [t_1,t_2]$ either $(\boldsymbol{z},k+i)\models\neg\phi_2$ holds or there exists $j \in [t_1,i-1]$  
	such that $(\boldsymbol{z},k+j)\models\neg\phi_1$. Since both $\neg\phi_1$ and $\neg\phi_2$ are sub-formulas of $\psi$, $(\boldsymbol{z},k+i)\models\neg\phi_2$ and  $(\boldsymbol{z},k+j)\models\neg\phi_1$ imply, respectively, that $\mathcal{C}_{k+j}^N(\neg\phi_1)$ and $\mathcal{C}_{k+i}^N(\neg\phi_2)$ are feasible. 
	We deduce that for all $i \in [t_1,t_2]$ either $b_{k+i}^N(\neg\phi_2)= 1$ or there exists $j \in [t_1,i-1]$ such that $b_{k+j}^N(\neg\phi_1) = 1$. 
	 Since the relation between $b_{k}^N(\psi)$, $b_{k+j}^N(\neg\phi_1)$, and $b_{k+i}^N(\neg\phi_2)$, as encoded in $\mathcal{C}_k^N(\psi)$, is
	\begin{equation}\label{eq:logic_relation2}
	b_k^N(\psi) = \wedge_{i=t_1}^{t_2} (b_{k+i}^N(\neg\phi_2)
	\vee_{j=t_1}^{i-1}b_{k+j}^N(\neg\phi_1)),
	\end{equation}
	we infer that $b_{k}^N(\psi) = 1$ is feasible. The feasibility of $\mathcal{C}_k^N(\psi)$ is then proved since  a feasible solution for $\mathcal{C}_k^N(\psi)$ can be obtained  by solving the conjunction of the constraints $\mathcal{C}_{k+j}^N(\neg\phi_1)  \setminus \{b_{k+j}^N(\neg\phi_1) = 1\}$ for all $j\in[t_1, t_2-1]$, $\mathcal{C}_{k+i}^N(\neg\phi_2) \setminus \{b_{k+i}^N(\neg\phi_2) = 1\}$ for all $j\in[t_1, t_2]$, constraint~\eqref{eq:logic_relation2}, and $b_{k}^N(\psi) = 1$.
\end{proof}


\section{Contract-Based Verification and Synthesis}\label{sec:contract_check}

We formulate verification and synthesis procedures that leverage under- and over-approximations of bounded StSTL contracts to solve Problem~\ref{prob:1}-\ref{prob:3} for the classes of stochastic systems introduced in Sec.~\ref{sec:handlingCons}.  
A first result provides  sound procedures to check contract consistency and compatibility (Problem~\ref{prob:1}).

\begin{thm}\label{thm:compati_consis}
Let  $\pfr{S}$ be a stochastic system belonging to one of the classes introduced in Sec.~\ref{sec:handlingCons} (Table~\ref{tab:chance_cons_formu}); let $C = (\phi_A,\phi_G)$ be an A/G contract where $\phi_A$ and
$\phi_G$ are bounded StSTL formulas over the system variables. If over- and under-approximations are available for both $\phi_A$ and $\neg \phi_A \vee \phi_G$, then the following hold:
	\begin{enumerate}
		\item If $\mathcal{C}_0^S(\phi_A)$ is feasible, then $C$ is compatible.
		\item If $\mathcal{C}_0^N(\phi_A)$ is infeasible, then $C$ is not
		compatible.
		\item If $\mathcal{C}_0^S(\neg \phi_A \vee \phi_G)$ is feasible, then $C$
		is consistent.
		\item If $\mathcal{C}_0^N(\neg \phi_A \vee \phi_G)$ is infeasible, then $C$
		is not consistent.
	\end{enumerate}
\end{thm}

%

\begin{proof}
	By Theorem~\ref{thm:suffi_encoding}, if $\mathcal{C}_0^S(\phi_A)$ is feasible, then $\phi_A$ is satisfiable, which indicates that $C$ is
	compatible. On the other hand, by Theorem \ref{thm:neces_encoding}, if $\mathcal{C}_0^N(\phi_A)$ is
	infeasible, then $\phi_A$ is unsatisfiable, hence $C$ is incompatible.
	The results on consistency can be obtained in the same way.
\end{proof}

The following result addresses refinement checking
(Problem~\ref{prob:2}). 

\begin{thm}\label{thm:refine}
Let  $\pfr{S}$ be a stochastic system belonging to one of the classes introduced in Sec.~\ref{sec:handlingCons} (Table~\ref{tab:chance_cons_formu}); let $C_1 = (\phi_{A1},\phi_{G1})$ and $C_2 = (\phi_{A2}, \phi_{G2})$ be A/G contracts whose assumptions and guarantees are 
bounded StSTL formulas over the system variables. If over- and under-approximations are available for $\psi_1 = \neg \phi_{A2} \vee \phi_{A1}$ and
	$\psi_2 = (\phi_{A1} \wedge \neg \phi_{G1}) \vee (\neg \phi_{A2} \vee \phi_{G2})$, then the following hold:
	\begin{enumerate}
		\item If $\mathcal{C}_0^N(\neg\psi_1)$ and $\mathcal{C}_0^N(\neg\psi_2)$ are
		infeasible, then $C_1 \preceq C_2$.
		\item If $\mathcal{C}_0^S(\neg\psi_1)$ or $\mathcal{C}_0^S(\neg\psi_2)$ are
		feasible, then $C_1 \not\preceq C_2$.
	\end{enumerate}
\end{thm}

\begin{proof}
The proof proceeds as in Theorem~\ref{thm:compati_consis}, by directly applying the definition of contract refinement. By Theorem~\ref{thm:neces_encoding}, if $\mathcal{C}_0^N(\neg\psi_1)$ and $\mathcal{C}_0^N(\neg\psi_2)$ are infeasible, then $\neg\psi_1$ and $\neg\psi_2$ are unsatisfiable, hence $\psi_1$ and $\psi_2$ are valid. We therefore obtain than $\phi_{A2} \rightarrow \phi_{A1}$ and $(\neg \phi_{A1} \vee \phi_{G1}) \rightarrow (\neg \phi_{A2} \vee \phi_{G2})$ are valid, hence $C_1 \preceq C_2$ by definition. 
Similarly, $\mathcal{C}_0^S(\neg\psi_1)$ or $\mathcal{C}_0^S(\neg\psi_2)$ being
		feasible implies that either $\psi_1$ or $\psi_2$ are not valid formulas by Theorem~\ref{thm:suffi_encoding}. We therefore conclude that $C_1 \not\preceq C_2$ holds. 
\end{proof}

The above decision procedures are not complete.
For instance, it is possible that $\mathcal{C}_0^S(\phi_A)$ is infeasible and $\mathcal{C}_0^N(\phi_A)$
is feasible, in which case we are not able to conclude on the satisfiability of $\phi_A$.
In this case, we increasingly refine piecewise-affine under- and over-approximations of chance constraints until we obtain an answer. 


Finally, as an application of Theorem~\ref{thm:compati_consis},  we provide a framework for the design of stochastic MPC schemes using StSTL contracts. We show how a stochastic optimization problem can be generated by enforcing contract consistency on the system in Fig.~\ref{fig:sys_structure} to obtain a control trajectory which solves Problem~\ref{prob:3}. 

\begin{example}[Generation of Stochastic MPC Schemes] \label{ex:mpc}
In stochastic MPC, the controller measures the plant state $x_k$ at time $k$ and derives a control input $u_k$ by solving a stochastic optimization problem. The plant state $x_{k+1}$ is a function of $u_k$  and the random external signal $w_k$ according to the system dynamics. 
Given a stochastic system described as in~\eqref{eq:sys}, where the environment input (disturbance) $w_k$ at each time $k$ follows a distribution $\mathcal{D}$, let the bounded StSTL contract  $C = (Q x_0 \leq  r,\phi)$ capture the system requirement that $\phi$ be satisfied if the initial state $x_0$ is in the polyhedron represented by set of linear inequalities $Q x_0 \leq  r$ for a fixed matrix $Q$ and vector $r$. 

Control synthesis can then be formulated as the problem of finding a control trajectory $\boldsymbol{u}$ that makes $C$ consistent and optimizes a predefined cost. For a finite horizon $H$, this translates into requiring that the guarantees of $C$ are satisfiable in the context of its assumptions, hence the conjunction of the following constraints
\begin{align*}
& (\boldsymbol{z}^H,0) \models (Q\bar{x}_0 \leq r) \to \phi, \; x_{k+1} = f(x_k,u_k,w_k), \\ & w_k \sim \mathcal{D},  x_0 = \bar{x}_0, 
\end{align*}
for $k=0,1,\ldots, H-1$, must be feasible, while optimizing a cost function $J_H(x_0, \boldsymbol{u}^H)$.
By calling $\psi := (Qx_0 \leq r) \to \phi$ and using Theorem~\ref{thm:compati_consis}, we can finally solve this problem using the under-approximation $\mathcal{C}_0^S(\psi)$ obtained as described in Sec. \ref{sec:encoding} over the horizon $H$, which provides the following stochastic optimization problem:
\begin{equation}\label{eq:MPC_opti1}
\begin{split}
\min_{\boldsymbol{u}^H} \quad J_H(x_0, \boldsymbol{u}^H), \quad
\mathrm{s.t.} \quad \mathcal{C}_0^S(\psi)
\end{split}
\end{equation}
to be executed in a receding horizon fashion. It is then possible to extend previous results on MPC from STL specifications~\cite{raman2014model} to stochastic linear systems. 
\end{example}


\section{Case Studies}\label{sec:sim_exam}

\begin{figure}[t]
	\centering
	\includegraphics[width=0.35\textwidth]{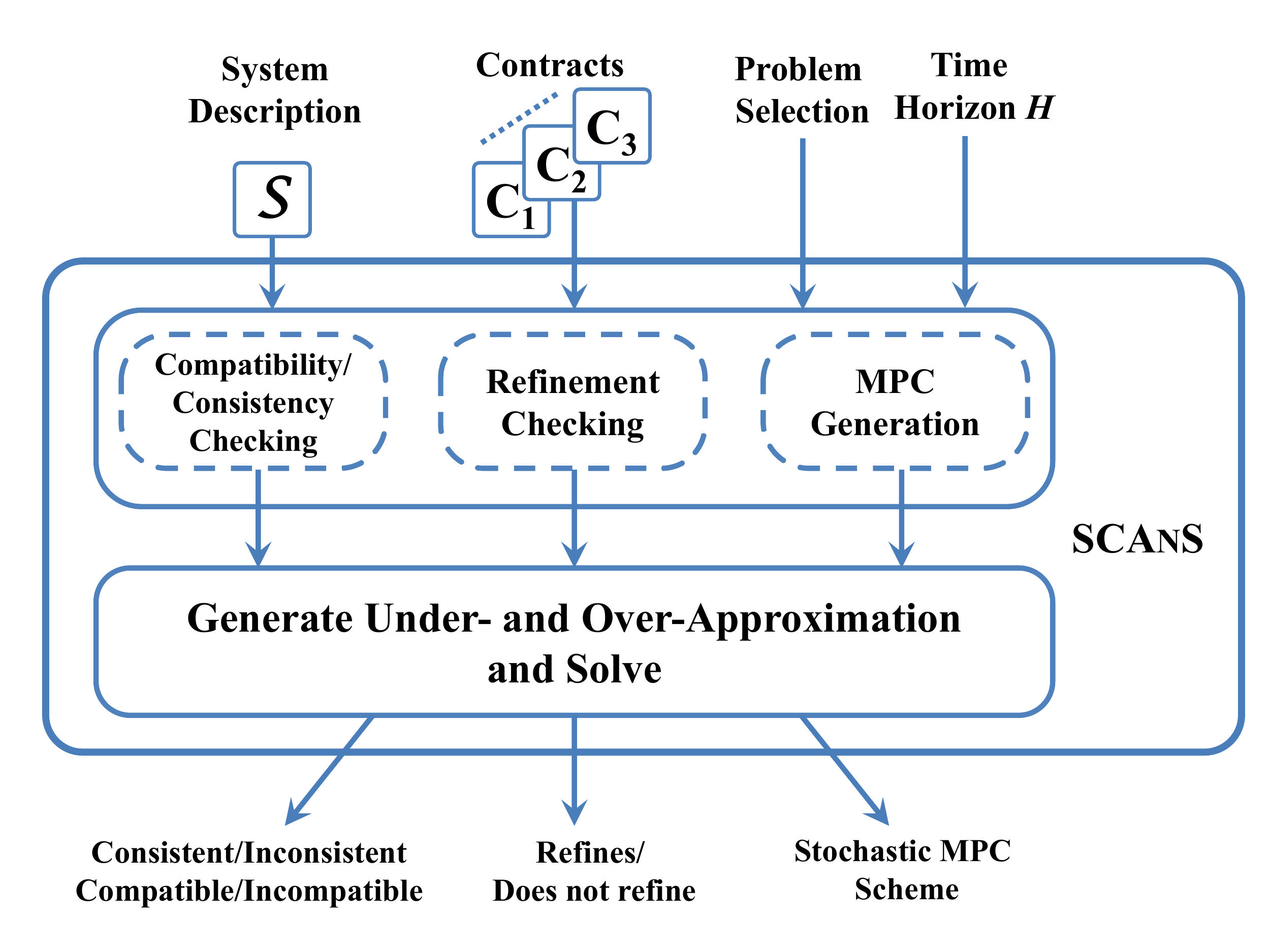}
	\caption{The \textsc{ScanS} Flow.}
	\label{fig:tool}
\end{figure}

We implemented the verification and synthesis procedures in Sec.~\ref{sec:contract_check} in the \textsc{Matlab} toolbox \textsc{SCAnS} (Stochastic Contract-based Analysis and Synthesis). 
As shown in Fig.~\ref{fig:tool}, \textsc{SCAnS} receives as inputs a system description in one of the classes of Sec.~\ref{sec:handlingCons}, a set of bounded StSTL contracts, a time horizon $H$, and a set of verification or synthesis tasks. In the verification flow, \textsc{SCAnS} computes under- and over-approximations of contract assumptions and guarantees and perform consistency, compatibility, and refinement checking of user-defined contracts using the results in Theorem~\ref{thm:compati_consis} and Theorem~\ref{thm:refine}. In the synthesis flow, \textsc{SCAnS} follows the procedure in Example~\ref{ex:mpc} to generate a stochastic optimization problem from a user-defined contract, which can be executed in a receding horizon scheme.   

We illustrate the effectiveness of our approach on two examples.
The first example utilizes both under- and over-approximations of StSTL formulas to perform contract compatibility, consistency, and refinement checking. 
The second example uses a formula under-approximation to synthesize an MPC controller for an aircraft power distribution network. \textsc{SCAnS} uses \textsc{Yalmip}~\cite{lofberg2004yalmip} to formulate mixed integer programs, \textsc{Gurobi}~\cite{gurobi} to solve mixed integer linear programs, and \textsc{bmibnb} (in \textsc{Yalmip}) to solve mixed integer nonlinear programs. All experiments ran on a $3.2$-GHz Intel Core i5 processor with $4$-GB memory.

\subsection{Contract-Based Verification}

We check compatibility and consistency for the contract and system in Example~\ref{sec:motiv_exmp}. By applying Theorem~\ref{thm:compati_consis} and the under-approximation in Sec.~\ref{sec:suffi_encode}, we find that $\mathcal{C}_0^S(\phi_{A1})$ is feasible, and so is $\mathcal{C}_0^S(\neg \phi_{A1} \vee \phi_{G1})$. Therefore, contract $(\phi_{A1},\phi_{G1})$ is both compatible and consistent.
Since the system is in the class of Sec.~\ref{sec:class1}, our encoding uses~\eqref{eq:linear_chance_cons_deter_suffi} and~\eqref{eq:linear_chance_cons_deter_neces}.
%
Given a contract $C_2$ defined as follows: 
\begin{equation*}
\begin{split}
\phi_{A2} &:= [1,0]x_0 \leq 3, \\
\phi_{G2} &:= \phi_{A2} \rightarrow \G_{[1,3]}\neg (\mathcal{P}\{[1,0]x_{2} \le 2\} \ge 0.6),
\end{split}
\end{equation*}
we can also check that $C_2 \preceq C_1$ by using the results in Theorem~\ref{thm:refine}. Moreover, to show the effectiveness of the proposed approximation, we increase the system dimension by redefining the dynamics as follows:
\begin{equation*}
\begin{split}
x_{k+1} &= A x_k + B_k u_k, \\
B_k &= I + 0.3\begin{bmatrix} w_{k,1} & & \\ & \ddots & \\ & & w_{k,1}\end{bmatrix}
-0.2\begin{bmatrix} & & w_{k,2} \\ & \iddots & \\ w_{k,2} & & \end{bmatrix}
\end{split}
\end{equation*}
where $A$ is a Jordan matrix constructed using blocks of dimension $2$ as in \eqref{eq:motivdyn}. Contract refinement checking on a system with $100$ state variables took about $20$ ms using the proposed approximate encoding, which is a $20\times$ reduction in execution time with respect to the exact encoding.

\subsection{Requirement Analysis and Control Synthesis for Aircraft Electric Power Distribution}

\begin{figure}[t]
	\centering
	\includegraphics[width=0.35\textwidth]{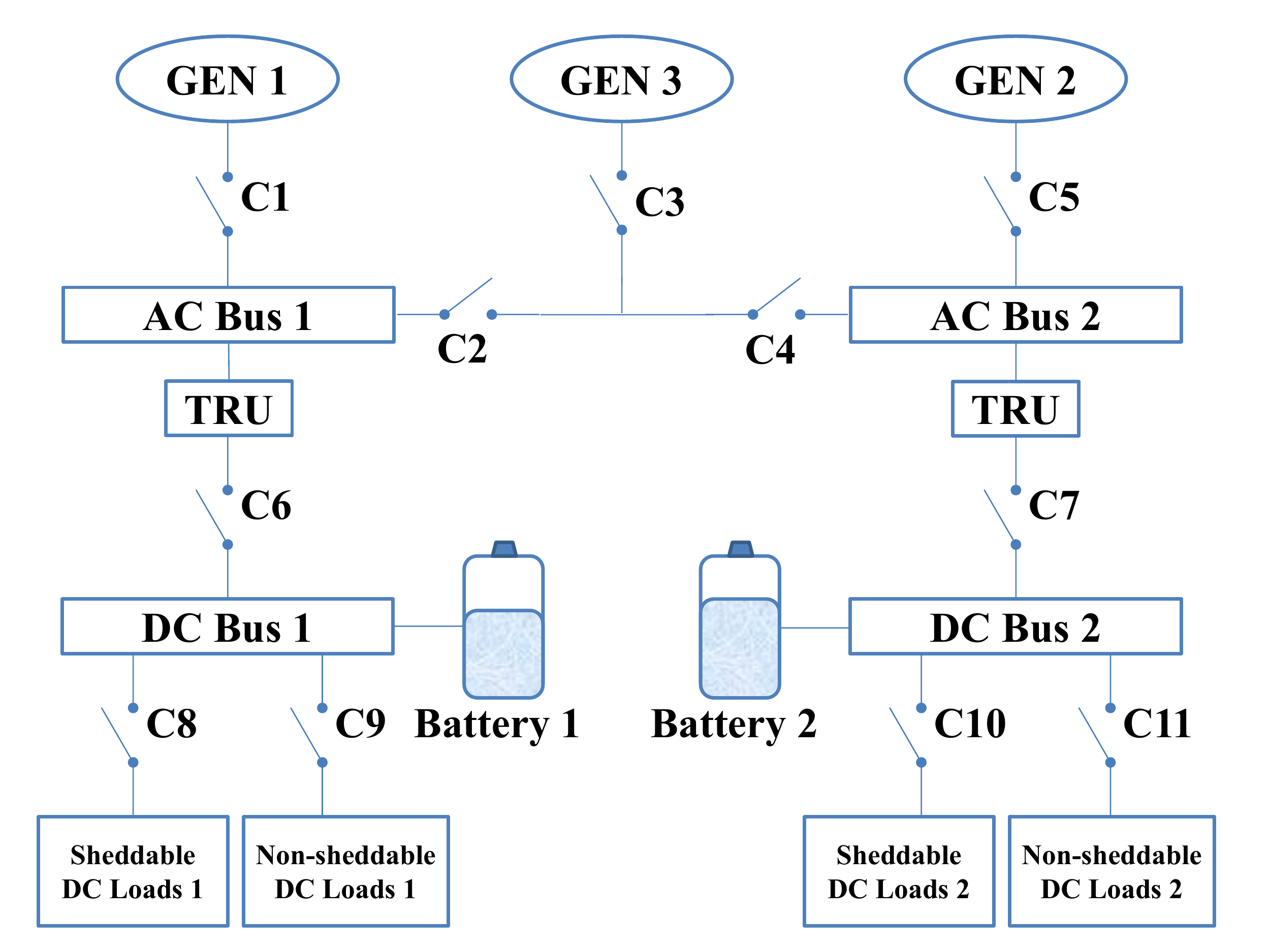}
	\caption{Simplified diagram of an aircraft power distribution system.}
	\label{fig:epsbus}
\end{figure}

An aircraft power system distributes power from
generators (engines) to loads by configuring a set of electronic control switches
denoted as contactors~\cite{nuzzo2014contract}. As shown in the simplified diagram of Fig.~\ref{fig:epsbus}, physical components of a power system include
generators, AC and DC buses, Transformer and Rectifier Units (TRUs), contactors (C1-C11), loads, and batteries. The controller, which is also denoted as Load Management System (LMS) and is not shown in the figure, determines the configuration of the contactors at each time instant, in order to provide the required power to the loads, while being subject to a set of constraints, e.g.,
on the battery charge level. 

A hierarchical LMS structure was proposed for aircraft power systems, which adopts two controller levels and is based on a deterministic model of the system~\cite{maasoumy2013optimal}. A high-level LMS (HL-LMS) operates at a lower frequency (e.g., 0.1~Hz) and provides
advice on the contactor configuration as obtained by
solving an optimization problem. The control objective is to provide power to the highest number of loads at each time (minimize load shedding) and reduce the switching frequency of contactors, hence the wear-and-tear associated with switching. A low-level LMS (LL-LMS), working at a
faster frequency (e.g., 1~Hz) takes critical decisions to place the system in safety mode by
shedding non-essential loads every time a generator fails. The LL-LMS accepts the suggestion of the HL-LMS only if it is safe.

We adopt the same model for the system architecture and the dynamics as in this reference design~\cite{maasoumy2013optimal}. The system state
is represented by the state of charge of the batteries, which are allowed to,
respectively, discharge or charge when the generator power
is insufficient or redundant with respect to the load power.
The system contains a number of generators $N_s = 3$ and a number of AC (DC) buses
$N_b = 2$, where each bus must be connected to a functional generator or TRU to receive
power. Each DC bus has $N_{sl} = 10$ sheddable loads and $N_{nsl} = 10$
non-sheddable loads, which are shown as lumped components in Fig.~\ref{fig:epsbus}. The maximum power supplied by the three generators is
$100$~kW (GEN1), $100$~kW (GEN2), and $85$~kW (GEN3). However, differently from the reference design~\cite{maasoumy2013optimal}, the power demand of each load is now a Gaussian random
variable. The average power demand assumes the values in Table II of our reference~\cite{maasoumy2013optimal}, while the variance is $0.1$ times larger than the average value. A controller based on stochastic MPC has been recently proposed for a similar power system model~\cite{shahsavari2015stochastic}. In this section, we show that \textsc{SCAnS} is able to \emph{automatically} design a controller that follows the same approach 
but can handle a \emph{richer} set of specifications.

We use StSTL to express the control specification $\psi$ for the HL-LMS, involving both
deterministic constraints on the network connectivity~\cite{maasoumy2013optimal} and stochastic constraints
on the battery levels. Sample requirements in $\psi$, over a time horizon of $20$ steps, are formalized as follows:
\begin{itemize}
	\item
	The battery charge level $B_j$ shall not be less than $0.3$ with
	probability larger than or equal to $0.95$, i.e.,
	\begin{equation}\label{batt_spec1}
	\square_{[1,20]} (0.3 - B_j)^{[0.95]}, \; j = 1,\ldots,N_b,
	\end{equation}
	\item
	If the battery level $B_j$ at time $0$ is less than or equal to $0.25$, then there exists a time in at most 5 steps at which $B_j$ equals or exceeds $0.4$ with probability larger than or equal to $0.95$, i.e., for all $j = 1,\ldots, N_b$:
	\begin{equation}\label{batt_spec2}
	(B_j - 0.25 \leq 0) \rightarrow \top \ \U_{[0,5]} (0.4 - B_j)^{[0.95]},
	\end{equation}
	\item
	If a generator is unhealthy, then it is disconnected from the buses.
	By denoting with
	$\boldsymbol{h} = (h_1,\ldots,h_{N_s})$
	the binary vector indicating the health status of the generators, where $1$ stands for ``healthy," and with
	$\boldsymbol{\delta}_j = [\delta_{1,j},\ldots,\delta_{N_s,j}]^T$ the vector whose component  $\delta_{i,j}$ is $1$ if and only if generator $i \in \{1,\ldots,N_s\}$ is connected to bus $j$, this requirement can be translated as
	\begin{equation}\label{batt_spec3}
	\square_{[0,20]} (\delta_{i,j} -{h}_i \leq 0), \qquad \forall \ i \in \{1, \ldots, N_s \}.
	\end{equation}
\end{itemize}
By calling $\psi$ the conjunction of all system requirement assertions, such as the ones above, the system-level contract is
\[
C_S = ( (\forall j \in \{1,\ldots,N_b\}\!\!: B_{j0} \in [0.2,1]) \wedge \textstyle\sum_{j=1}^{N_s} h_j \ge 2, \psi),
\]
stating that the specification $\psi$ must be satisfied if the initial battery
level is between $0.2$ and $1$ ($20\%$ and $100\%$ of the full level of charge) and if there are at least two healthy generators.

\textsc{SCAnS} was able to verify the consistency of $C_S$ using the result in Theorem~\ref{thm:compati_consis} and generate a
stochastic MPC scheme for the HL-LMS. We relied on the mixed integer linear under-approximation of $\psi$ into the constraint set
$\mathcal{C}_0^S(\psi)$ because of the large number of variables (more than $400$) in the optimization problems. When parsing $\psi$, deterministic constraints encoding the atomic propositions  $(0.3 - B_j)^{[0.95]}$
were formulated using~\eqref{eq:linear_chance_cons_deter_suffi}.
$\mathcal{C}_0^S(\psi)$ and the control objective
formed the optimization problem solved by the HL-LMS every $10$~s to provide suggestions to the LL-LMS. We observe that constraint~\eqref{batt_spec2}, capturing more complex transient behaviors, was not present in previous formulations~\cite{shahsavari2015stochastic}, while it could be easily expressed in StSTL and automatically accounted for in our MPC scheme.

In every simulation run, GEN2 is shut down at time $34$ to test the response of the LMS.
The contactor signals indicating the connection of the 3 generators to the 2 AC buses are in Fig.~\ref{fig:HL_LL_engine}. First, we observe that the LL-LMS connects GEN3 to
bus 2 at time $34$ to immediately replace the
faulty generator GEN2, before the HL-LMS can respond to this event at time 40. Meanwhile,
because the average total power consumption of
either bus 1 or bus 2 exceeds 85~kW (the maximum power supplied by GEN3), the LL-LMS sheds the loads at time $34$ in Fig.~\ref{fig:LL_bus2}. Conversely, the HL-LMS
does not detect this shutdown until time $40$. Once a new optimal configuration is
computed, as shown in Fig.~\ref{fig:HL_LL_engine}, the HL-LMS realizes that GEN2 must indeed be disconnected from bus 2 (requirement~\eqref{batt_spec3}) and proposes a configuration that connects GEN1 and GEN3 alternatively to the two buses. This prevents load shedding (all loads are now powered again) and better resource utilization, since the battery can now be effectively charged
when GEN1 is connected and then used to provide extra power when
GEN3 is connected. While the switching activity increases in this new configuration, the switching frequency is always compatible with the requirements and minimized by the MPC scheme.

\begin{figure}[t]
	\centering
	\includegraphics[width=0.3\textwidth]{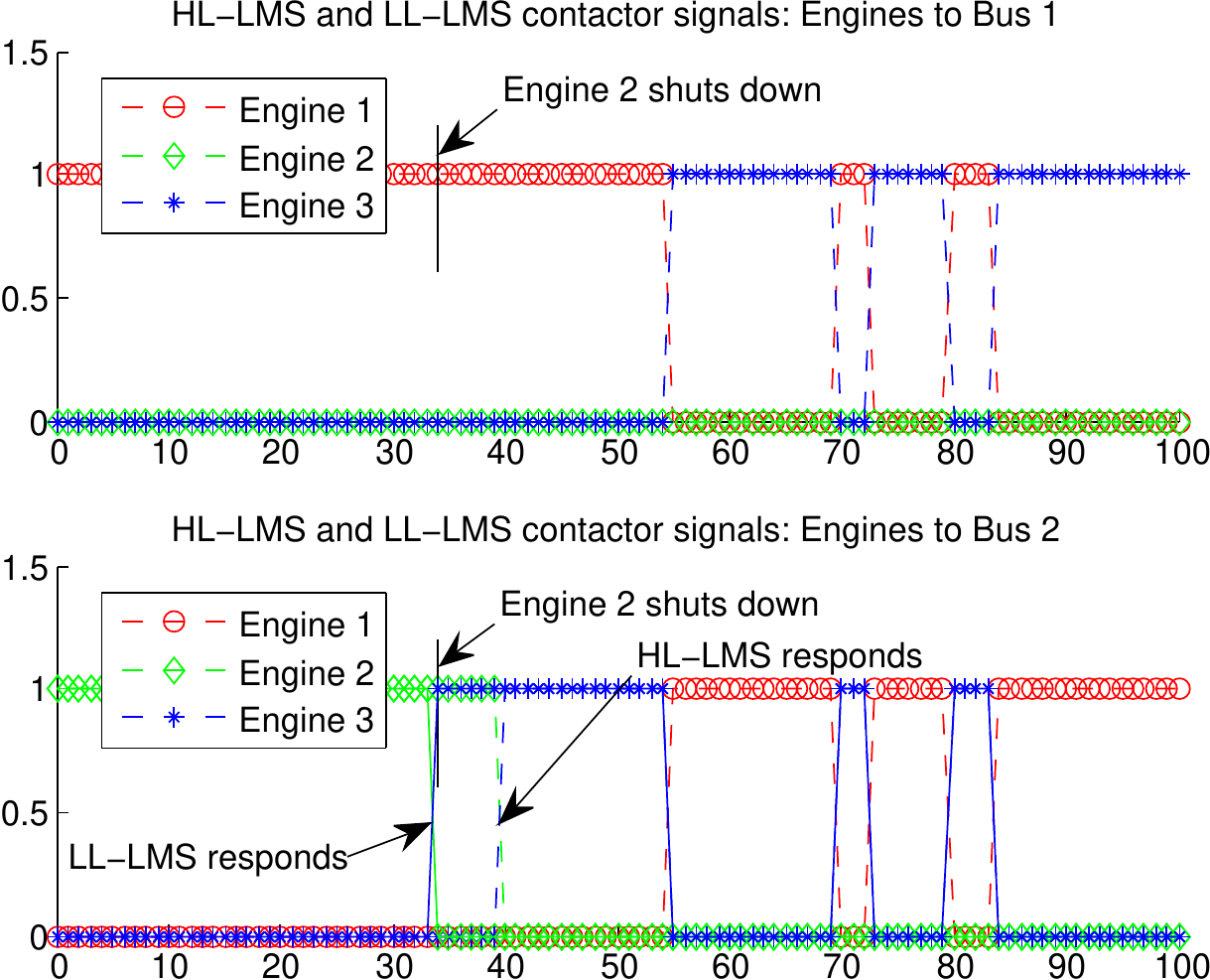}
	\caption{Contactor signals for the connection between generators (engines) and buses. The connection is present when the signal evaluates to $1$.}
	\label{fig:HL_LL_engine}
\end{figure}

\begin{figure}[t]
	\centering
	\includegraphics[width=0.3\textwidth]{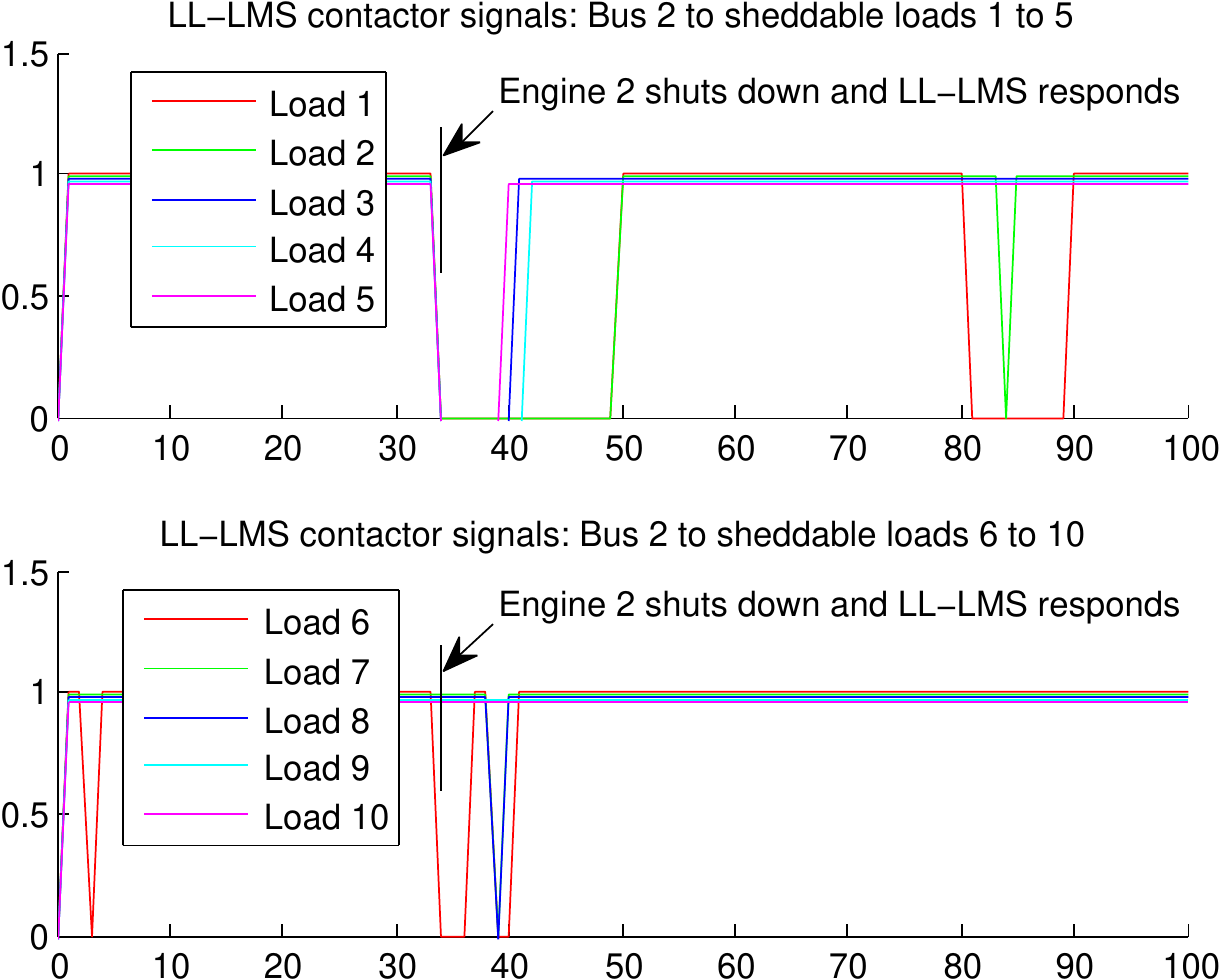}
	\caption{Contactor signals for the connection between sheddable loads and DC Bus 2. The connection is present when the signal evaluates to $1$.}
	\label{fig:LL_bus2}
\end{figure}

The trajectories of the battery charge level from 50 simulation runs
are shown in Fig.~\ref{fig:battery}. We see that the constraint~\eqref{batt_spec1} is
effective since the battery level mostly remains above $0.3$ after time $0$. Moreover,
most of the battery profiles starting from the initial condition
$B_{1,0} = B_{2,0} = 0.225$ climbs above $0.4$ before time $5$, which is consistent
with requirement \eqref{batt_spec2}. Finally, the rate of satisfaction of the constraint
$B_j \ge 0.3$, as estimated using 500 simulation runs, is larger than 0.95 at all times, which is consistent with requirement~\eqref{batt_spec1}.
One optimization run takes 0.05~s on average
and 0.24~s in the worst case.

\begin{figure}[t]
	\centering
	\includegraphics[width=0.3\textwidth]{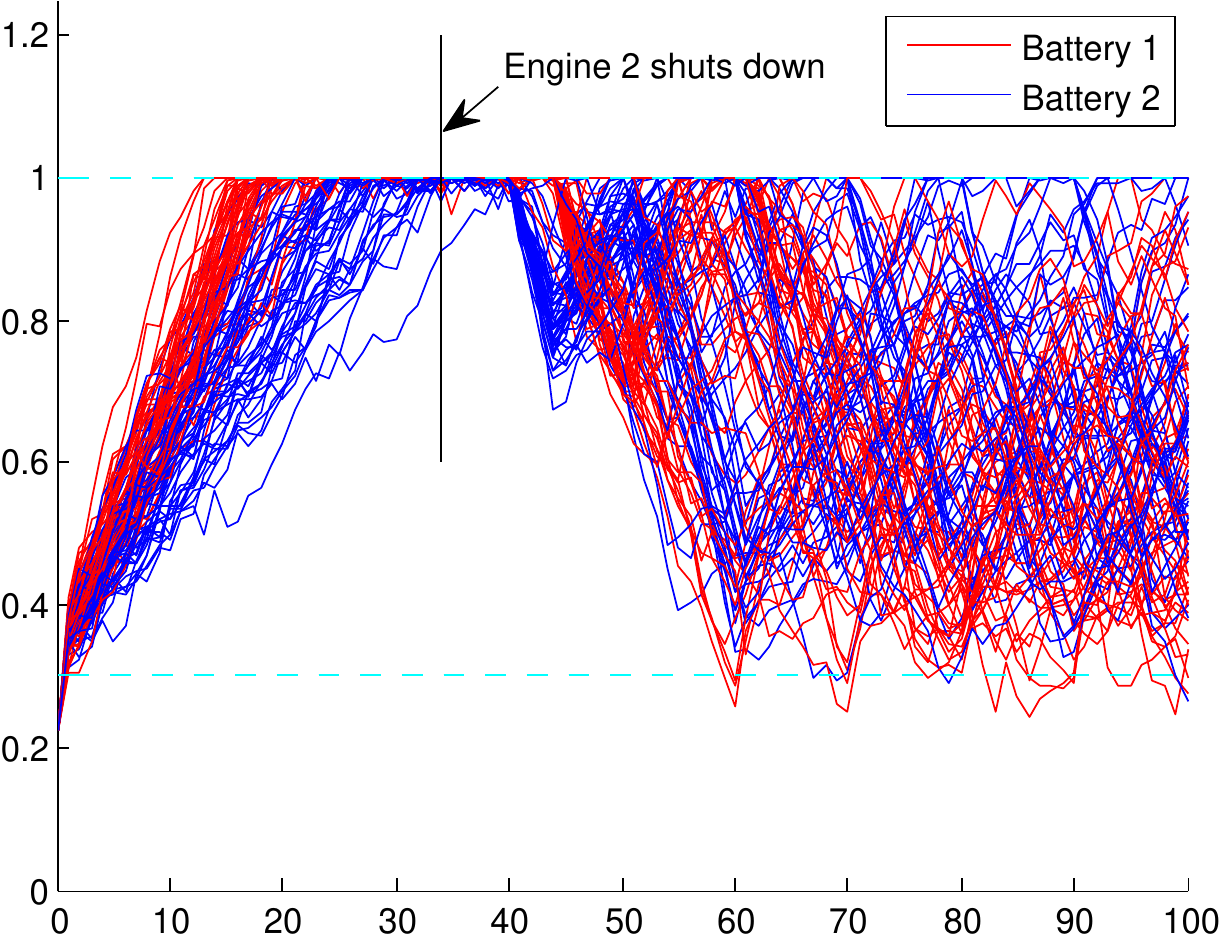}
	\caption{Battery charge level over time for 50 simulation runs.}
	\label{fig:battery}
\end{figure}

\section{Conclusions}
\label{sec:conclusions}

We developed an assume-guarantee contract framework and a supporting tool for the automated verification of certain classes of stochastic linear systems and the generation of stochastic Model Predictive Control (MPC) schemes. Our approach leverages Stochastic
Signal Temporal Logic to specify system behaviors and contracts, and algorithms that can efficiently encode and solve contract compatibility, consistency, and refinement checking problems using conservative approximations of probabilistic constraints. We illustrated the effectiveness of our approach on a few examples, including the control of aircraft electrical power distribution 
systems. Our tool can automatically design stochastic MPC schemes for a richer set of specifications than in previous work. Future work includes the investigation of mechanisms to improve the accuracy and scalability of our framework.

\bibliographystyle{IEEEtran}
{\small
	\bibliography{MyBib,nanoele-cps,IEEEproc2014,IEEEabrv}}

\begin{thebibliography}{10}
\providecommand{\url}[1]{#1}
\csname url@samestyle\endcsname
\providecommand{\newblock}{\relax}
\providecommand{\bibinfo}[2]{#2}
\providecommand{\BIBentrySTDinterwordspacing}{\spaceskip=0pt\relax}
\providecommand{\BIBentryALTinterwordstretchfactor}{4}
\providecommand{\BIBentryALTinterwordspacing}{\spaceskip=\fontdimen2\font plus
\BIBentryALTinterwordstretchfactor\fontdimen3\font minus
  \fontdimen4\font\relax}
\providecommand{\BIBforeignlanguage}[2]{{%
\expandafter\ifx\csname l@#1\endcsname\relax
\typeout{** WARNING: IEEEtran.bst: No hyphenation pattern has been}%
\typeout{** loaded for the language `#1'. Using the pattern for}%
\typeout{** the default language instead.}%
\else
\language=\csname l@#1\endcsname
\fi
#2}}
\providecommand{\BIBdecl}{\relax}
\BIBdecl

\bibitem{Benveniste2013}
A.~Benveniste, B.~Caillaud, D.~Nickovic, R.~Passerone, J.-B. Raclet,
  P.~Reinkemeier \emph{et~al.}, ``\BIBforeignlanguage{Anglais}{{Contracts for
  System Design}},'' INRIA, Rapport de recherche RR-8147, Nov. 2012.

\bibitem{Nuzzo15b}
P.~Nuzzo, A.~Sangiovanni-Vincentelli, D.~Bresolin, L.~Geretti, and T.~Villa,
  ``A platform-based design methodology with contracts and related tools for
  the design of cyber-physical systems,'' \emph{Proc. IEEE}, vol. 103, no.~11,
  Nov. 2015.

\bibitem{Benveniste08}
A.~Benveniste, B.~Caillaud, A.~Ferrari, L.~Mangeruca, R.~Passerone, and
  C.~Sofronis, ``Formal methods for components and objects.''\hskip 1em plus
  0.5em minus 0.4em\relax Berlin, Heidelberg: Springer-Verlag, 2008, ch.
  Multiple Viewpoint Contract-Based Specification and Design, pp. 200--225.

\bibitem{Alfaro01_2}
L.~de~Alfaro and T.~A. Henzinger, ``Interface automata,'' in \emph{Proc. Symp.
  Foundations of Software Engineering}.\hskip 1em plus 0.5em minus 0.4em\relax
  ACM Press, 2001, pp. 109--120.

\bibitem{KNP07a}
M.~Kwiatkowska, G.~Norman, and D.~Parker, ``Stochastic model checking,'' in
  \emph{Formal Methods for the Design of Computer, Communication and Software
  Systems: Performance Evaluation}, vol. 4486.\hskip 1em plus 0.5em minus
  0.4em\relax Springer, 2007, pp. 220--270.

\bibitem{KNP11}
------, ``{PRISM} 4.0: Verification of probabilistic real-time systems,'' in
  \emph{Proc. Int. Conf. Comput.-Aided Verification}, ser. LNCS,
  G.~Gopalakrishnan and S.~Qadeer, Eds., vol. 6806.\hskip 1em plus 0.5em minus
  0.4em\relax Springer, 2011, pp. 585--591.

\bibitem{gossler12}
G.~G\"{o}ssler, D.~N. Xu, and A.~Girault, ``Probabilistic contracts for
  component-based design,'' \emph{Formal Methods in System Design}, vol.~41,
  no.~2, pp. 211--231, 2012.

\bibitem{Caillaud10}
B.~Caillaud, B.~Delahaye, K.~Larsen, A.~Legay, M.~Pedersen, and A.~Wasowski,
  ``Compositional design methodology with {Constraint Markov Chains},'' in
  \emph{Int. Conf. Quantitative Evaluation of Systems}, Sep. 2010, pp.
  123--132.

\bibitem{delahaye2011abstract}
B.~Delahaye, J.-P. Katoen, K.~G. Larsen, A.~Legay, M.~L. Pedersen, F.~Sher, and
  A.~W{\k{a}}sowski, ``Abstract probabilistic automata,'' in \emph{Int.
  Workshop Verification, Model Checking, and Abstract Interpretation}.\hskip
  1em plus 0.5em minus 0.4em\relax Springer, 2011, pp. 324--339.

\bibitem{delahaye2011apac}
B.~Delahaye, K.~G. Larsen, A.~Legay, M.~L. Pedersen \emph{et~al.}, ``{APAC}: A
  tool for reasoning about abstract probabilistic automata,'' 2011.

\bibitem{MalerN04}
O.~Maler and D.~Nickovic, ``Monitoring temporal properties of continuous
  signals,'' in \emph{Formal Modeling and Analysis of Timed Systems}, 2004, pp.
  152--166.

\bibitem{nemirovski2006convex}
A.~Nemirovski and A.~Shapiro, ``Convex approximations of chance constrained
  programs,'' \emph{SIAM Journal on Optimization}, vol.~17, no.~4, pp.
  969--996, 2006.

\bibitem{delahaye2010probabilistic}
B.~Delahaye, B.~Caillaud, and A.~Legay, ``Probabilistic contracts: A
  compositional reasoning methodology for the design of stochastic systems,''
  in \emph{Int. Conf. Application of Concurrency to System Design}, 2010, pp.
  223--232.

\bibitem{hansson1994logic}
H.~Hansson and B.~Jonsson, ``A logic for reasoning about time and
  reliability,'' \emph{Formal aspects of computing}, vol.~6, no.~5, pp.
  512--535, 1994.

\bibitem{clarke1986automatic}
E.~M. Clarke, E.~A. Emerson, and A.~P. Sistla, ``Automatic verification of
  finite-state concurrent systems using temporal logic specifications,''
  \emph{ACM Transactions on Programming Languages and Systems (TOPLAS)},
  vol.~8, no.~2, pp. 244--263, 1986.

\bibitem{sadigh2016}
D.~Sadigh and A.~Kapoor, ``Safe control under uncertainty with probabilistic
  signal temporal logic,'' in \emph{Proceedings of Robotics: Science and
  Systems}, ser. RSS '16, 2016.

\bibitem{Clark99}
E.~M. Clarke, O.~Grumberg, and D.~A. Peled, \emph{Model Checking}.\hskip 1em
  plus 0.5em minus 0.4em\relax Cambridge, MA: The {MIT} Press, 2008.

\bibitem{Sangiovanni-Vincentelli2012a}
A.~Sangiovanni-Vincentelli, W.~Damm, and R.~Passerone, ``{Taming Dr.
  Frankenstein: Contract-Based Design for Cyber-Physical Systems},''
  \emph{European Journal of Control}, vol. 18-3, no.~3, pp. 217--238, 2012.

\bibitem{Durrett10}
R.~Durrett, \emph{Probability: theory and examples}.\hskip 1em plus 0.5em minus
  0.4em\relax Cambridge university press, 2010.

\bibitem{raman2014model}
V.~Raman, A.~Donz{\'e}, M.~Maasoumy, R.~M. Murray, A.~Sangiovanni-Vincentelli,
  and S.~A. Seshia, ``Model predictive control with signal temporal logic
  specifications,'' in \emph{Proc. Int. Conf. Decision and Control}.\hskip 1em
  plus 0.5em minus 0.4em\relax IEEE, 2014, pp. 81--87.

\bibitem{bradley1977applied}
S.~Bradley, A.~Hax, and T.~Magnanti, \emph{Applied mathematical
  programming}.\hskip 1em plus 0.5em minus 0.4em\relax Addison Wesley, 1977.

\bibitem{harris1998signal}
C.~M. Harris and D.~M. Wolpert, ``Signal-dependent noise determines motor
  planning,'' \emph{Nature}, vol. 394, no. 6695, pp. 780--784, 1998.

\bibitem{elia2005remote}
N.~Elia, ``Remote stabilization over fading channels,'' \emph{Systems \&
  Control Letters}, vol.~54, no.~3, pp. 237--249, 2005.

\bibitem{de2006mode}
C.~E. de~Souza, A.~Trofino, and K.~A. Barbosa, ``Mode-independent filters for
  {M}arkovian jump linear systems,'' \emph{IEEE Trans. Automatic Control},
  vol.~51, no.~11, pp. 1837--1841, 2006.

\bibitem{l2008operations}
W.~L. Winston, \emph{Operations Research: Applications \& Algorithms}.\hskip
  1em plus 0.5em minus 0.4em\relax Thomson Business Press, 2008.

\bibitem{lofberg2004yalmip}
J.~L\"{o}fberg, ``{YALMIP}: A toolbox for modeling and optimization in
  {MATLAB},'' in \emph{Proc. CACSD Conference}, Taipei, Taiwan, 2004.

\bibitem{gurobi}
\BIBentryALTinterwordspacing
I.~Gurobi~Optimization, ``Gurobi optimizer reference manual,'' 2015. [Online].
  Available: \url{http://www.gurobi.com}
\BIBentrySTDinterwordspacing

\bibitem{nuzzo2014contract}
P.~Nuzzo, H.~Xu, N.~Ozay, J.~B. Finn, A.~L. Sangiovanni-Vincentelli, R.~M.
  Murray, A.~Donz{\'e}, and S.~A. Seshia, ``A contract-based methodology for
  aircraft electric power system design,'' \emph{IEEE Access}, vol.~2, pp.
  1--25, 2014.

\bibitem{maasoumy2013optimal}
M.~Maasoumy, P.~Nuzzo, F.~Iandola, M.~Kamgarpour, A.~Sangiovanni-Vincentelli,
  and C.~Tomlin, ``Optimal load management system for aircraft electric power
  distribution,'' in \emph{Proc. Int. Conf. Decision and Control}.\hskip 1em
  plus 0.5em minus 0.4em\relax IEEE, 2013, pp. 2939--2945.

\bibitem{shahsavari2015stochastic}
B.~Shahsavari, M.~Maasoumy, A.~Sangiovanni-Vincentelli, and R.~Horowitz,
  ``Stochastic model predictive control design for load management system of
  aircraft electrical power distribution,'' in \emph{Proc. American Control
  Conference}, 2015, pp. 3649--3655.

\end{thebibliography}



\end{document}